\documentclass[runningheads]{llncs}

\usepackage[T1]{fontenc}

\usepackage{mathtools,bbold,stmaryrd,tikz,amsmath,amssymb,enumerate,hyperref,cleveref}
\usetikzlibrary{cd}

\usepackage{color}

\SetSymbolFont{stmry}{bold}{U}{stmry}{m}{n}

\newcommand{\Freyd}{\ensuremath{\mathbf{Freyd}}}
\newcommand{\Sub}{\ensuremath{\mathbf{Subset}}}
\newcommand{\Set}{\ensuremath{\mathbf{Set}}}
\newcommand{\Law}{\ensuremath{\mathbf{Law}}}
\newcommand{\Lab}{\ensuremath{\mathbf{Label}}}
\newcommand{\Prof}{\ensuremath{\mathbf{Prof}}}

\newcommand{\cat}[1]{\ensuremath{\mathbf{#1}}}
\newcommand{\vfreyd}[1]{\ensuremath{\mathbf{#1}}}
\newcommand{\op}{\ensuremath{^{\mathrm{op}}}}
\newcommand{\id}[1][]{\ensuremath{\mathrm{id}_{#1}}}
\newcommand{\Day}{\ensuremath{_{\mathrm{Day}}}}
\newcommand{\Kl}{\ensuremath{\mathrm{Kl}}}

\newcommand{\In}{\ensuremath{\inplus}}

\newcommand{\mzero}{*}
\newcommand{\ezero}{J}
\newcommand{\mone}{\circ}
\newcommand{\eone}{I}

\newcommand{\ehom}{\underline{\hom}}
\newcommand{\epcomp}{\mathbin{\hat{\mone}}}

\DeclareFontFamily{U}{min}{}
\DeclareFontShape{U}{min}{m}{n}{<-> udmj30}{}

\tikzset{global scale/.style={
    scale=#1,
    every node/.style={scale=#1}
  }
}

\crefname{axiom}{axiom}{axioms}

\title{Duoidally enriched Freyd categories\thanks{Jesse Sigal is partly funded by Huawei.}}

\author{Chris Heunen\orcidID{0000-0001-7393-2640} \and
Jesse Sigal\orcidID{0000-0002-5117-8752}}

\authorrunning{C. Heunen and J. Sigal}

\institute{School of Informatics, University of Edinburgh, United Kingdom,\\
\email{\{chris.heunen, jesse.sigal\}@ed.ac.uk}
}

\begin{document}
\maketitle

\begin{abstract}
  Freyd categories provide a semantics for first-order effectful programming languages by capturing the two different orders of evaluation for products.
  We enrich Freyd categories in a duoidal category, which provides a new, third choice of parallel composition.
  Duoidal categories have two monoidal structures which account for the sequential and parallel compositions.
  The traditional setting is recovered as a full coreflective subcategory for a judicious choice of duoidal category.
  We give several worked examples of this uniform framework, including the parameterised state monad, basic separation semantics for resources, and interesting cases of change of enrichment.

  \keywords{Freyd category \and duoidal category \and Kleisli category \and Lawvere theory \and monad}
\end{abstract}

\section{Introduction}\label{sec:intro}

Computational effects encapsulate interactions of a computer program with its environment in a modular way, and are a staple of modern programming languages~\cite{plotkinpower:effects}.
Originally captured by strong monads~\cite{moggi:monads}, they have been extended to Arrows to deal with input as well as output~\cite{jacobsheunenhasuo:arrows}, to Lawvere theories to better combine effects algebraically~\cite{staton:freyd}, to PROs and PROPs to deal with non-cartesian settings~\cite{lack:props}, and to Freyd categories to deal with effects that are not higher-order~\cite{levypowerthielecke:environments}. 

Freyd categories let one compose effectful computations both in sequence and, to some extent, in parallel, and reason about such compositions rigorously.
For an effectful computation $f \colon a \to b$, we may embed it, the domain, and the codomain into a larger context by extending with $- \otimes c$ for any object $c$ and monoidal-like operation $\otimes$, which we write as $f \otimes \id \colon a \otimes c \to b \otimes c$.
Intuitively, $f \otimes \id$ does not interact with $c$.
Effectful computations need not commute as they may alter the environment: $(f \otimes \id). (\id \otimes g) \neq (\id \otimes g). (f \otimes \id)$ in general.

But what if we want to track more data about computations than just types and effects? For example, suppose we want to annotate every computation with its resource needs: there could \textit{e.g.}\ be a set $R$ of resources, and every computation $f$ requires a certain subset $P \subseteq R$ of resources for it to execute.
Sequencing two computations needs all resources to execute both, so if $f \colon a \to b$ and $g \colon b \to c$ require resources $P$ and $Q$ respectively, then $g . f$ requires $P \cup Q$. 
The same is true for parallel composition: if $f_1 \colon a_1 \to b_1$ and $f_2 \colon a_2 \to b_2$ require $P_1$ and $P_2$ respectively, then $f_1 \otimes f_2 \colon a_1 \otimes a_2 \to b_1 \otimes b_2$ requires $P_1 \cup P_2$.
However, it is often desirable to restrict $P_1$ and $P_2$ by requiring $P_1 \cap P_2 = \emptyset$ so that morphisms composed in parallel use different resources.
If we have an identity map $\id \colon a \to a$ for all $a$ which requires $\emptyset \subseteq R$, then we can always form $f \otimes \id$ for any $f$, but what of the general case?

This article proposes a solution that achieves just this: enrich Freyd categories in duoidal categories.
Duoidal categories carry two interacting monoidal structures that will account for the sequential and parallel composition of both the effectful computations and the extra data we want to track, such as the resources above.
We provide a concrete example for resources in \Cref{sub:state-monoid}.

\Cref{sec:enrichedfreyd} introduces duoidally enriched Freyd categories.
\Cref{sec:examples} shows the breadth of such categories by treating disparate examples: separation semantics for resources as above, indexed state monads, and Kleisli categories of Lawvere theories.
\Cref{sec:adjunction} shows that a judicious choice of duoidal enriching category recovers traditional Freyd categories as a full coreflective subcategory, and \Cref{sec:characterisation} gives an abstract characterisation of duoidally enriched Freyd categories in purely algebraic terms. \Cref{sec:changeofenrichment} considers changing the enriching duoidal category, accounting for \textit{e.g.}\ changing the underlying permission model in the example above. \Cref{sec:conclusion} concludes and suggests directions for future work.

\paragraph*{Related work} 

Morrison and Penneys define a $\cat{V}$-monoidal category~\cite{morrisonpenneys:braided} for braided monoidal $\cat{V}$ as a $\cat{V}$-category with parallel composition that interacts well with the braid. In the case $\cat{V}$ is braided (and thus duoidal), our definition of a $\cat{V}$-Freyd category is similar.
However, we also require bifunctorality of the hom objects, an important difference for some of our constructions.

The abstract characterisation in \Cref{sec:characterisation} is inspired by Fujii's characterisation of PROs and PROPs~\cite{fujii:unified} as monoids in $\cat{MonCat_{lax}}\big( \cat{N}\op \times \cat{N}, \Set \big)$ and $\cat{MonCat_{lax}}\big( \cat{P}\op \times \cat{P}, \Set \big)$ respectively, where $\cat{N}$ and $\cat{P}$ have natural numbers as objects and equalities respectively bijections as morphisms.

Garner and L\'opez Franco describe a general framework for commutativity using categories enriched in the sequential product of a duoidal category \cite{garner_commutativity_2016}.
Their framework requires the duoidal category to be \emph{normal}, meaning that the two units are isomorphic.
Only with this requirement and others do they define a monoidal structure on their category of enriched categories, and do not define a monoidal enriched category.
We do not require normality.

Finally, Forcey~\cite{forcey_enrichment_2004}, and Batanin and Markl~\cite{batanin_centers_2012} enrich over duoidal categories, but using the parallel product instead.
We choose to enrich over the sequential product in order to define examples in which this is the appropriate choice.

\section{Duoidally enriched Freyd categories}\label{sec:enrichedfreyd}

This section introduces duoidally enriched Freyd categories (in \Cref{subsec:definition}), but first we discuss Freyd categories (in \Cref{subsec:freydcats}) and duoidal categories (in \Cref{subsec:duoidal-categories}).

\subsection{Freyd categories}\label{subsec:freydcats}

Freyd categories provide semantics for first-order call-by-value programming languages with effects~\cite{staton:freyd}.
We will generalise the definition of a Freyd category slightly so that the effect free fragment need not have products, beginning with the following preliminary definitions~\cite{levypowerthielecke:environments,powerrobinson:premonoidal}.

\begin{definition}\label{def:binoidal}
  A category $\cat{C}$ is \emph{binoidal} when it comes with endofunctors $(-) \ltimes x$ and $x \rtimes (-)$ for each object $x$ such that $x \ltimes y = x \rtimes y$ for all $y$; write $x \otimes y$ for this object.
  A morphism $f \colon x \to y$ is \emph{central} if for any morphism $g \colon x' \to y'$ the two maps $(y \rtimes g) . (f \ltimes x')$ and $(f \ltimes y') . (x \rtimes g)$ of type $x \otimes x' \to y \otimes y'$ are equal, as are the two maps $(y' \rtimes f) . (g \ltimes x)$ and $(g \ltimes y) . (x' \rtimes f)$ of type $x' \otimes x \to y' \otimes y$.
  Central morphisms form a wide subcategory $Z(\cat{C})$ called the \emph{centre}.
\end{definition}

\begin{definition}\label{def:premonoidal}
  A binoidal category $\cat{C}$ is \emph{premonoidal} when equipped with an object $e$ and families of central isomorphisms $\alpha \colon (x \otimes y) \otimes z \to x \otimes (y \otimes z)$, $\lambda \colon e \otimes x \to x$, and $\rho \colon x \otimes e \to x$ that are natural in each component and satisfy triangle and pentagon equations.
\end{definition}

\begin{definition}\label{def:premonoidal-functors}
  A functor $F \colon \cat{C} \to \cat{D}$ between premonoidal categories is a \emph{premonoidal functor} when equipped with central morphisms $\eta \colon e_{\cat{D}} \to F\left(e_{\cat{C}}\right)$ and $\mu \colon F(x) \otimes_{\cat{D}} F(y) \to F\left( x \otimes_{\cat{C}} y\right)$ such that $\mu$ is natural in each component, and the following diagrams commute:
  \[\begin{tikzcd}[column sep=small, global scale=0.8]
    {(F(x) \otimes_\cat{D} F(y)) \otimes_\cat{D} F(z)} & {F(x) \otimes_\cat{D} (F(y) \otimes_\cat{D} F(z))} \\
    {F(x \otimes_\cat{C} y) \otimes_\cat{D} F(z)} & {F(x) \otimes_\cat{D} F(y \otimes_\cat{C} z)} \\
    {F((x \otimes_\cat{C} y) \otimes_\cat{C} z)} & {F(x \otimes_\cat{C} (y \otimes_\cat{C} z))}
    \arrow["{\mu \otimes \id}"', from=1-1, to=2-1]
    \arrow["\mu"', from=2-1, to=3-1]
    \arrow["F \alpha_\cat{C}"',from=3-1, to=3-2]
    \arrow["\alpha_\cat{D}" {yshift=3pt}, from=1-1, to=1-2]
    \arrow["{\id \otimes \mu}", from=1-2, to=2-2]
    \arrow["\mu", from=2-2, to=3-2]
  \end{tikzcd}
  \qquad
  \begin{tikzcd}[column sep=small, global scale=0.8]
    {e_\cat{D} \otimes_\cat{D} F(x)} & {F(e_\cat{C}) \otimes_\cat{D} F(x)} \\
    {F(x)} & {F(e_\cat{C} \otimes_\cat{C} x)} \\
    {F(x) \otimes_\cat{D} e_\cat{D}} & {F(x) \otimes_\cat{D} F(e_\cat{C}) } \\
    {F(x)} & {F(x \otimes_\cat{C} e_\cat{C})}
    \arrow["\lambda_\cat{D}"', from=1-1, to=2-1]
    \arrow["{\eta \otimes \id}", from=1-1, to=1-2]
    \arrow["\mu", from=1-2, to=2-2]
    \arrow["F \lambda_\cat{C}", from=2-2, to=2-1]
    \arrow["\rho_\cat{D}"', from=3-1, to=4-1]
    \arrow["{\id \otimes \eta}", from=3-1, to=3-2]
    \arrow["\mu", from=3-2, to=4-2]
    \arrow["F \rho_\cat{C}", from=4-2, to=4-1]
  \end{tikzcd}\]
  A premonoidal functor is \emph{strong (strict)} when $\eta$ and $\mu$ are isomorphisms (identities).
\end{definition}

Note that a strict premonoidal functor $F$ preserves associators and unitors on the nose.
Recall that a functor $F \colon \cat{C} \to \cat{D}$ between monoidal categories is \emph{lax monoidal} when it comes with a morphism $\eta \colon I \to F(I)$ and a natural transformation $\mu \colon F(X) \otimes F(Y) \to F(X \otimes Y)$ satisfying coherence conditions. It is \emph{strong monoidal} when $\eta$ and $\mu$ are invertible. Lax/strong monoidal functors are closed under composition.
Here now is our definition of a Freyd category.

\begin{definition}\label{def:freyd-category}
  A \emph{Freyd category} consists of a monoidal category $\cat{M}$ and a premonoidal category $\cat{C}$ with the same objects, and an identity-on-objects strict premonoidal functor $J \colon \cat{M} \to \cat{C}$ whose image lies in $Z(\cat{C})$.
  A morphism $J \to J'$ of Freyd categories consists of a strong monoidal functor $F_0 \colon \cat{M} \to \cat{M}'$ and a strong premonoidal functor $F_1 \colon \cat{C} \to \cat{C}'$ such that $F_1 J = J' F_0$.
  Freyd categories and their morphisms form a category $\Freyd$.
\end{definition}

\subsection{Duoidal categories}\label{subsec:duoidal-categories}

A duoidal category carries two interacting monoidal structures, that one may intuitively  think of as sequential and parallel composition, but let us give the definition~\cite[Definition~6.1]{aguiar2010monoidal} before examples.

\begin{definition}\label{def:duoidal}
  A category $\cat{V}$ is \emph{duoidal} when it comes with two monoidal structures $(\cat{V}, \mzero, \ezero)$ and $(\cat{V}, \mone, \eone)$, a natural transformation $\zeta_{A,B,C,D} \colon (A \mone B) \mzero (C \mone D) \to (A \mzero C) \mone (B \mzero D)$, and three morphisms $\Delta \colon \ezero \to \ezero \mone \ezero$, $\nabla \colon \eone \mzero \eone \to \eone$, and $\epsilon \colon \ezero \to \eone$ such that $(\eone, \nabla, \epsilon)$ is a monoid in $(\cat{V}, \mzero, \ezero)$ and $(\ezero, \Delta, \epsilon)$ is a comonoid in $(\cat{V}, \mone, \eone)$, and the following diagrams commute:
\[\begin{tikzcd}[column sep=small, global scale=0.8]
  {((A \mone B) \mzero (C \mone D)) \mzero(E \mone F)} & {(A \mone B) \mzero((C \mone D) \mzero(E \mone F))} \\
  {((A \mzero C) \mone (B \mzero D)) \mzero(E \mone F)} & {(A \mone B) \mzero ((C \mzero E) \mone (D \mzero F))} \\
  {((A \mzero C) \mzero E) \mone ((B \mzero D) \mzero F)} & {(A \mzero (C \mzero E)) \mone (B \mzero (D \mzero F))}
  \arrow["{\zeta \mzero \id}"', from=1-1, to=2-1]
  \arrow["\zeta"', from=2-1, to=3-1]
  \arrow["\alpha", from=1-1, to=1-2]
  \arrow["{\id \mzero \zeta}", from=1-2, to=2-2]
  \arrow["\zeta", from=2-2, to=3-2]
  \arrow["{\alpha \mone \alpha}"' {yshift=-3pt}, from=3-1, to=3-2]
\end{tikzcd}
\begin{tikzcd}[column sep=tiny, global scale=0.8]
  {\ezero \mzero (A \mone B)} & {(\ezero \mone \ezero) \mzero (A \mone B)} \\
  {A \mone B} & {(\ezero \mzero A) \mone (\ezero \mzero B)} \\
  {(A \mone B) \mzero \ezero} & { (A \mone B) \mzero (\ezero \mone \ezero)} \\
  {A \mone B} & {(A \mzero \ezero) \mone (B \mzero \ezero)}
  \arrow["\lambda"', from=1-1, to=2-1]
  \arrow["\Delta \mzero \id" {yshift=3pt}, from=1-1, to=1-2]
  \arrow["\zeta", from=1-2, to=2-2]
  \arrow["\lambda \mone \lambda", from=2-2, to=2-1]
  \arrow["\rho"', from=3-1, to=4-1]
  \arrow["\id \mzero \Delta" {yshift=3pt}, from=3-1, to=3-2]
  \arrow["\zeta", from=3-2, to=4-2]
  \arrow["\rho \mone \rho", from=4-2, to=4-1]
\end{tikzcd}\]
\[\begin{tikzcd}[column sep=small, global scale=0.8]
  {((A \mone B) \mone C) \mzero ((D \mone E) \mone F)} & {(A \mone (B \mone C)) \mzero (D \mone (E \mone F))} \\
  {((A \mone B) \mzero (D \mone E)) \mone (C \mzero F)} & {(A \mzero D) \mone ((B \mone C) \mzero (E \mone F))} \\
  {((A \mzero D) \mone (B \mzero E)) \mone (C \mzero F)} & {(A \mzero D) \mone ((B \mzero E) \mone (C \mzero F))}
  \arrow["\zeta"', from=1-1, to=2-1]
  \arrow["{\zeta \mone \id}"', from=2-1, to=3-1]
  \arrow["\alpha"', from=3-1, to=3-2]
  \arrow["{\alpha \mzero \alpha}" {yshift=3pt}, from=1-1, to=1-2]
  \arrow["\zeta", from=1-2, to=2-2]
  \arrow["{\id \mone \zeta}", from=2-2, to=3-2]
\end{tikzcd}
\begin{tikzcd}[column sep=tiny, global scale=0.8]
  {\eone \mone (A \mzero B)} & {(\eone \mzero \eone) \mone (A \mzero B)} \\
  {A \mzero B} & {(\eone \mone A) \mzero (\eone \mone B)} \\
  {(A \mzero B) \mone \eone} & { (A \mzero B) \mone (\eone \mzero \eone)} \\
  {A \mzero B} & {(A \mone \eone) \mzero (B \mone \eone)}
  \arrow["\lambda"', from=1-1, to=2-1]
  \arrow["\nabla \mone \id"' {yshift=3pt}, from=1-2, to=1-1]
  \arrow["\zeta"', from=2-2, to=1-2]
  \arrow["\lambda \mzero \lambda", from=2-2, to=2-1]
  \arrow["\rho"', from=3-1, to=4-1]
  \arrow["\id \mone \nabla"' {yshift=3pt}, from=3-2, to=3-1]
  \arrow["\zeta"', from=4-2, to=3-2]
  \arrow["\rho \mzero \rho", from=4-2, to=4-1]
\end{tikzcd}\]
  We may write $\left(\cat{V}, \mzero, \ezero, \mone, \eone\right)$ or $\left(\cat{V}, \mzero, \mone\right)$ to be explicit about the role of each monoidal structure.
\end{definition}

\begin{example}\label{ex:duoidal-braided}
  Any braided monoidal category becomes duoidal by letting both monoidal structures coincide and $\zeta$ be the middle-four interchange $x \otimes y \otimes z \otimes w \allowbreak \to x \otimes z \otimes y \otimes w$ up to associativity. In particular, any symmetric or cartesian monoidal category is duoidal~\cite[Proposition~6.10, Example~6.19]{aguiar2010monoidal}.
\end{example}

\begin{example}\label{ex:duoidal-opposite}
  If $\left(\cat{V}, \mzero, \ezero, \mone, \eone\right)$ is duoidal, so is $\left(\cat{V}\op, \mone, \eone, \mzero, \ezero\right)$, with opposite structure maps~\cite[Section~6.1.2]{aguiar2010monoidal}.
\end{example}

\begin{example}\label{ex:duoidal-products}
  If $\left(\cat{V}, \otimes, I\right)$ is a monoidal category with products, $\left(\cat{V}, \otimes, I, \times, 1\right)$ is duoidal with $\zeta = \left\langle \pi_1 \otimes \pi_1, \pi_2 \otimes \pi_2\right\rangle$, $\Delta = \left\langle\id, \id\right\rangle$, and $\nabla$ and $\epsilon$ terminal maps.
  Similarly, if a monoidal category $\cat{V}$ has coproducts, $\left(\cat{V}, +, 0, \otimes, I\right)$ is duoidal~\cite[Example~6.19]{aguiar2010monoidal}.
\end{example}

\begin{example}\label{ex:duoidal-presheaf}
  If $\left(\cat{V}, \mzero, \ezero, \mone, \eone\right)$ is small and duoidal, straightforward calculation shows Day convolution~\cite{day:convolution} of each monoidal structure makes the  category of presheaves $([\cat{V}\op, \Set], \mzero\Day, \cat{V}(-, \ezero), \mone\Day, \cat{V}(-, \eone))$ again duoidal where
  \begin{displaymath}
    \left(F \mzero\Day G\right)\left(A\right) = \int^{B, C} \cat{V}\left(A, B \mzero C\right) \times F\left(B\right) \times G\left(C\right)
  \end{displaymath}
  and likewise for $\mone\Day$.
  An analogous construction holds for $\left[\cat{V}, \Set\right]$ by starting with $\cat{V}\op$.
\end{example}

\begin{example}\label{ex:duoidal-comp-day}
  An endofunctor on $\Set$ is finitary when it preserves filtered colimits and is therefore determined on finite sets. 
  Finitary endofunctors are closed under functor composition, $\circ$, with unit $\mathrm{Id}$; closed under Day convolution with products, $\times\Day$, with unit $\Set\left(1, -\right) \cong \mathrm{Id}$; making $\big([\Set, \Set]_f, \times\Day, \mathrm{Id}, \circ, \mathrm{Id}\big)$ a duoidal category. \cite{garner_commutativity_2016}
\end{example}

\begin{example}\label{ex:duoidal-profunctor}
  For a small monoidal category $(\cat{M}, \oplus, e)$, the category of $\Set$-valued endoprofunctors $\Prof(\cat{M}) \coloneqq \left[\cat{M}\op \!\times\! \cat{M}, \Set\right]$ is duoidal $(\Prof(\cat{M}), \oplus\Day,\! \diamond)$ with profunctor composition $(P \diamond Q)(a, c) \coloneqq \int^{b} P(a, b) \times Q(b, c)$ (having unit $\cat{M}(-, -)$) and Day convolution of $\oplus$ on both sides $(P \oplus\Day Q)(a, b) \coloneqq \int^{a_1, a_2, b_2, b_2} \cat{M}(a, a_1 \oplus a_2) \times \cat{M}(b_1 \oplus b_2, b) \times P(a_1, b_1) \times Q(a_2, b_2)$ (having unit $\cat{M}(-, e) \times \cat{M}(e, -)$). \cite{garner_commutativity_2016}
\end{example}

\begin{example}\label{ex:duoidal-subset}
  An important example for us is the category $\Sub$ of \emph{distinguished subsets}.
  Objects are pairs of sets $\left(X, A\right)$ such that $X \subseteq A$ and morphisms $f \colon (X, A) \to (Y, B)$ are functions $f \colon A \to B$ with $f(X) \subseteq Y$.
  We call $X$ the \emph{distinguished subset}.
  Composition and identities are as in $\Set$.
  We may suppress the distinguished subset $X$ by writing $a \In A$ when $a \in X$.
  Next, we give two monoidal structures on $\Sub$.

  The first is the cartesian product: $(X, A) \times (Y, B) \coloneqq (X \times Y, A \times B)$ on objects, and $f \times g$ as in $\Set$ on morphisms, with unit $(1, 1)$.
  Associators and unitors are as in $\Set$.
  This is also a categorical product.

  The second is the \emph{disjunctive product}: on objects $(X, A) \otimes (Y, B)$ is defined as $\big(X \times Y, (A \times Y) \cup (X \times B)\big)$ with unit $(1, 1)$.
  We again have $f \times g$ on morphisms, which is well-defined.
  Finally, the coherence maps are restricted versions of those for the cartesian product.

  Now $\big(\Sub, \otimes, (1, 1), \times, (1, 1)\big)$ is duoidal by \Cref{ex:duoidal-products}: $\Delta$ and $\nabla$ are unitors, $\epsilon$ is the identity, and $\zeta \colon \big((X, A) \times (Y, B)\big) \otimes \big((Z, C) \times (W, D)\big) \to \big((X, A) \otimes (Z, C)\big) \times \big((Y, B) \otimes (W, D)\big)$ is the restricted middle-four interchange; all axioms are inherited from $(\Set, \times, 1)$ via \Cref{ex:duoidal-braided}.

  The important difference between $\big(\Sub, \otimes, \times\big)$ and $(\Set, \times, \times)$ is that $\zeta$ is not invertible in the former (as it is not surjective as a $\Set$ map).
  This allows Freyd categories enriched in $\Sub$ a premonoidal-like structure.
\end{example}

\subsection{Concrete definition}\label{subsec:definition}

We are now ready for the titular notion of this paper. We first give a concrete definition, leaving an abstract characterisation to Section~\ref{sec:characterisation}.

\begin{definition}\label{def:v-freyd}
  Let $\left(\cat{V}, \mzero, \ezero, \mone, \eone\right)$ be a duoidal category and $(\cat{M}, \oplus, e)$ a monoidal category.
  A \emph{$\cat{V}$-Freyd category over $\cat{M}$} consists of 
  \begin{itemize}
    \item a bifunctor $\vfreyd{C} \colon \cat{M}\op \times \cat{M} \to \cat{V}$
    \item an extranatural family $\mathsf{idt} \colon \eone \to \vfreyd{C}(a, a)$, meaning $\vfreyd{C}(\id, f) . \mathsf{idt} = \vfreyd{C}(f, \id) . \mathsf{idt}$ 
    \item an extranatural family $\mathsf{seq} \colon \vfreyd{C}(a, b) \mone \vfreyd{C}(b, c) \to \vfreyd{C}(a, c)$, meaning $\mathsf{seq}$ is natural in $a$ and $c$, and $\mathsf{seq} . (\id \mone \vfreyd{C}(f, \id)) = \mathsf{seq} . (\vfreyd{C}(\id, f) \mone \id)$
    \item a morphism $\mathsf{zero} \colon \ezero \to \vfreyd{C}(e, e)$
    \item a natural family $\mathsf{par} \colon \vfreyd{C}(a_1, b_1) \mzero \vfreyd{C}(a_2, b_2) \to \vfreyd{C}(a_1 \oplus a_2, b_1 \oplus b_2)$
  \end{itemize}
  satisfying the following axioms:
  \begin{romanenumerate}
    \item\label[axiom]{ax:mon.mor.id} $\mathsf{idt}$ is the identity for $\mathsf{seq}$, that is, $\mathsf{seq} . (\mathsf{idt} \mone \id) = \lambda$ and symmetrically;
    \item\label[axiom]{ax:mon.mor.ass} $\mathsf{seq}$ is associative, that is, $\mathsf{seq} . (\mathsf{seq} \mone \id) = \mathsf{seq} . (\id \mone \mathsf{seq}) . \alpha$;
    \item\label[axiom]{ax:obj.lax.id} $\mathsf{zero}$ is the identity for $\mathsf{par}$, that is, $\vfreyd{C}(\lambda^{-1}, \lambda) . \mathsf{par} . (\mathsf{zero} \mzero \id) = \lambda$ and symmetrically;
    \item\label[axiom]{ax:obj.lax.ass} $\mathsf{par}$ is associative, that is, $\vfreyd{C}(\alpha^{-1}, \alpha) . \mathsf{par} . (\mathsf{par} \mzero \id) = \mathsf{par} . (\id \mzero \mathsf{par}) . \alpha$;
    \item\label[axiom]{ax:mon.mor.e.eta} $\mathsf{idt}$ respects $\mathsf{zero}$ via $\mathsf{idt} . \epsilon = \mathsf{zero}$;
    \item\label[axiom]{ax:mon.mor.e.mu} $\mathsf{idt}$ respects $\mathsf{par}$ via $\mathsf{idt} . \nabla = \mathsf{par} . ( \mathsf{idt} \mzero \mathsf{idt} )$;
    \item\label[axiom]{ax:mon.mor.m.eta} $\mathsf{seq}$ respects $\mathsf{zero}$ via $\mathsf{seq} . ( \mathsf{zero} \mone \mathsf{zero} ). \Delta = \mathsf{zero}$;
    \item\label[axiom]{ax:mon.mor.m.mu} $\mathsf{seq}$ respects $\mathsf{par}$ via $\mathsf{seq} . ( \mathsf{par} \mone \mathsf{par} ) . \zeta = \mathsf{par} . ( \mathsf{seq} \mzero \mathsf{seq} )\text.$
  \end{romanenumerate}
  See \Cref{sec:diagrams:vfreyd} for diagrams expressing the axioms.
\end{definition}

\begin{definition}\label{def:v-freydmorphism}
  A \emph{morphism of $\cat{V}$-Freyd categories} consists of a strong monoidal functor $F_0 \colon \cat{M} \to \cat{M}'$ and a natural transformation $F_1 \colon \vfreyd{C}(a, b) \to \vfreyd{C}'\left(F_0 a, F_0 b\right)$ satisfying:
  \begin{itemize}
    \item $F_1 . \mathsf{idt} = \mathsf{idt}'$;
    \item $F_1 . \mathsf{seq} = \mathsf{seq}' . \left(F_1 \mone F_1\right)$;
    \item $\cat{C}'\left(\id, \mu\right) . \mathsf{par}' . \left(F_1 \mzero F_1\right) = \cat{C}'\left(\mu, \id\right) . F_1 . \mathsf{par}$.
  \end{itemize}
$\cat{V}$-Freyd categories and morphisms between them form a category $\cat{V}\text{-}\Freyd$.
\end{definition}

Our definition differs from the duoidally enriched categories of Batanin and Markl \cite{batanin_centers_2012} in a few important ways.
They use $\mzero$ for sequencing and $\mone$ for parallel composition.
Their analogues to \cref{ax:mon.mor.e.eta,ax:mon.mor.e.mu,ax:mon.mor.m.eta,ax:mon.mor.m.mu} are $\mathsf{idt} = \mathsf{zero} . \epsilon$, $\mathsf{idt} = \mathsf{par} . ( \mathsf{idt} \mone \mathsf{idt} ) . \Delta$, $\mathsf{seq} . ( \mathsf{zero} \mzero \mathsf{zero} ) = \mathsf{zero} . \nabla$, and $\mathsf{seq} . ( \mathsf{par} \mzero \mathsf{par} ) = \mathsf{par} . ( \mathsf{seq} \mone \mathsf{seq} ) . \zeta\text.$
Additionally, their monoidal structure is more enriched while we inherit ours from a $\Set$-category, namely $\cat{M}$.
Thus, we believe both notions are not inter-expressible.

\section{Examples}\label{sec:examples}

This section works out three applications of duoidally enriched Freyd categories: resource management (in \Cref{sub:state-monoid}), indexed state (in \Cref{subsec:indexed-state}), and Kleisli categories of Lawvere theories (in \Cref{sub:kleislilawvere}).

\subsection{Stateful functions and separated monoids}\label{sub:state-monoid}

To deal with resources abstractly, we first introduce the novel notion of a separated monoid.

\begin{definition}\label{def:sep-monoid}
  A monoid $(M,\bullet,e)$ is \emph{separated} when it comes with a binary relation $\Vert$ such that:
  $e \Vert m$ and $m \Vert e$;
  and $mm' \Vert n$ iff $m \Vert n$ and $m' \Vert n$;
  and $m \Vert nn'$ iff $m \Vert n$ and $m \Vert n'$.
\end{definition}

Examples include $(\mathbb{N}, +, 0)$ with $x \Vert y$ iff $x = 0$ or $y = 0$; finite subsets $(\mathcal{P}_f(R), \cup, \emptyset)$ of a fixed set $R$, with $P \Vert Q$ iff $P \cap Q = \emptyset$; and products of separated monoids under pointwise separation.
Separated monoids parametrise duoidal categories of resources as follows.

\begin{definition}\label{def:labelled-set}
  Let $(M, \Vert)$ be a separated monoid.
  The category $\Lab_M$ of \emph{$M$-labelled sets} has as objects functions $\ell \colon A \to M$ and as morphisms functions $f \colon A \to A'$ with $\ell' f = \ell$.
%
  This category has a monoidal structure $\bullet$ as follows: 
  on objects, $\ell \bullet \ell' \colon A \times A' \to M$ sends $(a, a')$ to $\ell(a) \bullet \ell'(a')$; on morphisms, $f \bullet f' = f \times f'$;
  the unit $\mathsf{cst}_e \colon 1 \to M$ picks out $e \in M$.
%
  There is a second monoidal structure $\Vert$ as follows:
  on objects, $\ell \Vert \ell'$ is the restriction of $\ell \bullet \ell'$ to $\{(a,a') \mid \ell(a) \Vert \ell'(a') \}$; on morphisms, $f \Vert f'=f \times f'$.
%
  The category $\left(\Lab_M, \Vert, \mathsf{cst}_e, \bullet, \mathsf{cst}_e\right)$ is duoidal with $\zeta \colon \left(\ell_1 \bullet \ell'_1\right) \Vert \left(\ell_2 \bullet \ell'_2\right) \to \left(\ell_1 \Vert \ell_2\right) \bullet \left(\ell'_1 \Vert \ell'_2\right)$ the restricted version of the $\zeta$ for $\left(\Set, \times, 1, \times, 1\right)$.
\end{definition}

Think of objects in $\Lab_M$ as sets of elements labelled with their resource needs.
The multiplication of $M$ combines resources, and the separation $\Vert$ relates non-conflicting resources.
We will now describe an enriched Freyd category where morphisms are labelled by resources as in the introduction.

Fix a countable family $R = \{ x, y, z, \ldots \}$ of sets which we think of as resources.
The set $\mathcal{P}_f(R)$ of finite subsets of $R$ is a monoid under union, and becomes a separated monoid under disjointness.
For set of resources $Q \in \mathcal{P}_f(R)$, fix a product of sets $\Pi_{x \in Q} x \eqqcolon \Pi_Q$ which thus combines the resources in $Q$.
Write $\pi_{Q'} \colon \Pi_{Q} \to \Pi_{Q'}$ for the projection if $Q' \subseteq Q$, and given a map $f \colon a \times \Pi_{Q'} \to b \times \Pi_{Q'}$ for sets $a$ and $b$, write $f_{Q'}^Q$ for the map $a \times \Pi_Q \to b \times \Pi_Q$ induced by $f$ when $Q' \subseteq Q$ which leaves the extra resources $Q \setminus Q'$ unchanged.

We will define a $\Lab_{\mathcal{P}_f(R)}$-Freyd category over $\Set$ of state-transforming functions.
Let $\vfreyd{C}(a,b)$ be the function from the disjoint union of $\Set(a \times \Pi_Q, b \times \Pi_Q)$ over $Q \in \mathcal{P}_f(R)$ to $\mathcal{P}_f(R)$, that sends $f \colon a \times \Pi_Q \to b \times \Pi_Q$ to $Q$. 
Thus, a map $f \in \vfreyd{C}(a, b)$ with label $Q$ is an effectful computation from $a$ to $b$ which can effect only resources in $Q$.
This becomes a bifunctor under pre- and post-composition.
Writing $\cup$ for $\bullet$ and $\cap$ for $\Vert$ for the sake of concreteness, the structure maps are:
\begin{align*}
  \mathsf{idt} \colon \mathsf{cst}_\emptyset &\to \vfreyd{C}(a, a)
  & \mathsf{zero} \colon \mathsf{cst}_\emptyset &\to \vfreyd{C}(1, 1) \\
  \star &\mapsto (\emptyset, \id{}_{a \times 1})
  & \star &\mapsto (\emptyset, \id{})
\end{align*}
\vspace*{-1.75\baselineskip} 
\begin{align*}
  \mathsf{seq} \colon \vfreyd{C}(a, b) \cup \vfreyd{C}(b, c) &\to \vfreyd{C}(a, c) \\
  \left((P, f), (Q, g)\right) &\mapsto \left(P \cup Q, g_{Q}^{P \cup Q} . f_{P}^{P \cup Q}\right)
\end{align*}
\vspace*{-1.75\baselineskip} 
\begin{align*}
  \mathsf{par} \colon \vfreyd{C}(a, b) \cap \vfreyd{C}(a', b') &\to \vfreyd{C}(a \times a', b \times b') \\
  ((Q, f), (Q', f')) &\mapsto \\
  &\hspace*{-26pt}\left(Q \cup Q', \left(\id \times \left\langle \pi_{Q}, \pi_{Q'}\right\rangle^{-1} \right) m^{-1} . (f \times f') . m . \left(\id \times \left\langle \pi_{Q}, \pi_{Q'}\right\rangle\right)\right)
\end{align*}
where $\left\langle \pi_{Q}, \pi_{Q'} \right\rangle \colon \Pi_{Q \cup Q'} \to \Pi_{Q} \times \Pi_{Q'}$ is invertible because $Q \cap Q' = \emptyset$ and $m$ is middle-four interchange.
So $\mathsf{par}$ places maps in parallel up to rearranging state.

\subsection{Indexed state}\label{subsec:indexed-state}

An important computational effect is global state. However, it is often inflexible as the type of storage remains constant over time. In this example the type can vary.
We use the duoidal category of finitary endofunctors on $\Set$ of \Cref{ex:duoidal-comp-day} to give a $\left[\Set, \Set\right]_f$-Freyd category over $\Set$ based on the state monad $(s \times ( - ))^s$, extending Atkey's example \cite{kurz_algebras_2009}.
Define $\vfreyd{C}(a, b) = (b \times ( -))^{a}$, which is a bifunctor via pre- and post-composition.
The natural structure maps are:
\begin{align*}
  \mathsf{idt}_{X} \colon X &\to \left(a \times X \right)^{a}
  & \mathsf{zero}_X \colon X &\to \left(1 \times X \right)^{1}\\
  x &\mapsto \lambda a . (x, a)
  & x &\mapsto \lambda \star . (x, \star)
\end{align*}
\vspace*{-1.75\baselineskip}
\begin{align*}
  \mathsf{seq}_{X} \colon \left(b \times \left( \left(c \times X \right)^{b} \right)\right)^{a} &\to \left(c \times X \right)^{a} \\
  f &\mapsto \mathrm{eval} . f 
\end{align*}
\vspace*{-1.75\baselineskip}
\begin{align*}
  \mathsf{par}_X \colon \textstyle\int^{Y,Z} X^{Y \times Z} \times \left(b \times Y \right)^{a} \times \left(c \times Z \right)^{a'} &\to ((b \times c) \times X )^{a \times a'} \\
  (k, f, g) &\mapsto (\id \times k) . m . (f \times g)
\end{align*}
where $\mathrm{eval} \colon b \times \left(c \times X \right)^{b} \to c \times X$ is the evaluation map and $m$ is the middle-four interchange.
$\mathsf{idt}$ and $\mathsf{seq}$ are the unit and multiplication of a state monad but with varying types of state.

\subsection{Kleisli categories of Lawvere theories}\label{sub:kleislilawvere}

Lawvere theories model effectful computations. Functional programmers might be more familiar with Kleisli categories of monads, to which they are closely related. Here we describe an indexed version, which models independent effects in parallel. 
Let $\Law$ be the category of Lawvere theories. Its initial object is the theory $\mathcal{S}$ of sets, the unit for the \emph{tensor product} $\otimes$ of Lawvere theories~\cite{hylandplotkinpower:sumtensor}. This makes $\Law$ a symmetric monoidal category, with the special property that there exist inclusion maps $\phi_i \colon \mathcal{L}_i \to \mathcal{L}_1 \otimes \mathcal{L}_2$.
Thus the functor category $\left[\Law, \Set\right]$ is monoidal under Day convolution with unit the constant functor $\Law(\mathcal{S}, -) \simeq \mathbb{1}$.
As this category also has products, \Cref{ex:duoidal-products} makes it duoidal.

Now, $\Law$ is equivalent to the category of finitary monads~\cite[Chapter~3]{adamekrosicky:presentable}: any Lawvere theory $\mathcal{L}$ induces a monad $T(\mathcal{L})$, and any map $\theta$ of Lawvere theories induces a monad morphism $T(\theta)$.
Every monad $T$ on $\Set$ is canonically bistrong: there are maps $\mathrm{st}_T \colon a \times T b \to T (a \times b)$ and $\mathrm{st'}_T \colon T a \times b \to T (a \times b)$ making the two induced maps $(a \times T b) \times c \to T ((a \times b) \times c)$ equal.
Each monad morphism $T(\theta)$ preserves strength: $T(\theta)_{a \times b} . \mathrm{st}_{T(\mathcal{L})} = \mathrm{st}_{T(\mathcal{L}')} . ( \id \times T(\theta)_{b})$.

We now show a $[\Law, \Set]$-Freyd category over $\Set$ given by the Kleisli construction on Lawvere theories.
Define on objects $\vfreyd{C}(a, b) = T( - )(b)^a$, and on morphisms $\vfreyd{C}(f, g) \colon \vfreyd{C}(a, b) \Rightarrow \vfreyd{C}(a', b')$ by $\vfreyd{C}(f, g)_{\mathcal{L}}(k) =T(\mathcal{L})(g) . k . f$, finally:
\begin{align*}
  \mathsf{idt}_{\mathcal{L}} \colon 1 &\to T(\mathcal{L})(a)^a
  & \mathsf{zero}_{\mathcal{L}} \colon 1 &\to T(\mathcal{L})(1)^1 \\
  \star &\mapsto \eta
  & \star &\mapsto \eta
\end{align*}
\vspace*{-1.75\baselineskip}
\begin{align*}
  \mathsf{seq}_{\mathcal{L}} \colon T(\mathcal{L})(b)^a \times T(\mathcal{L})(c)^b &\to T(\mathcal{L})(c)^a \\
  (f, g) &\mapsto \mu . T(\mathcal{L})g . f
\end{align*}
\vspace*{-1.75\baselineskip}
\begin{align*}
  \mathsf{par}_{\mathcal{L}} \colon \textstyle\int^{\mathcal{L}_1, \mathcal{L}_2} \Law(\mathcal{L}_1\!\otimes\!\mathcal{L}_2, \mathcal{L}) \!\times\! T(\mathcal{L}_1)(b_1)^{a_1} \!\times\! T(\mathcal{L}_2)(b_2)^{a_2} \to T(\mathcal{L})(b_1 \!\times\! b_2)^{a_1 \times a_2} \\
  \left(\theta, f_1, f_2\right) \mapsto T(\theta) . \mu . T(\mathcal{L}_1 \otimes \mathcal{L}_2)(\mathrm{st'}) . \mathrm{st} . \left(T(\phi_1) \times T(\phi_2)\right) . \left(f_1 \times f_2\right)
\end{align*}
Intuitively, $\mathsf{par}$ lets us put Kleisli maps in parallel as long as their effects are forced to commute (by $\otimes$).
So $\mathsf{idt}_{\mathcal{L}}$ and $\mathsf{seq}_{\mathcal{L}}$ are the identity and composition for the Kleisli category of $T(\mathcal{L})$.
The definition of $\mathsf{par}_{\mathcal{L}}$ seems noncanonical because of the use of $T(\mathcal{L}_1 \otimes \mathcal{L}_2)(\mathrm{st'}) . \mathrm{st}$, but it is not:
$\mu . T(\mathcal{L}_1 \otimes \mathcal{L}_2)(\mathrm{st'}) . \mathrm{st} . \left(T(\phi_1) \times T(\phi_2)\right)$ and $\mu . T(\mathcal{L}_1 \otimes \mathcal{L}_2)(\mathrm{st}) . \mathrm{st'} . \left(T(\phi_1) \times T(\phi_2)\right)$ are equal by definition of $\mathcal{L}_1 \otimes \mathcal{L}_2$.

\section{Adjunction between \texorpdfstring{$\Sub\text{-}\Freyd$}{Sub-Freyd} and \texorpdfstring{$\Freyd$}{Freyd}}\label{sec:adjunction}

Now let us explain how $\cat{V}$-Freyd categories generalise Freyd categories.
Our approach is similar to Power's \cite{power_premonoidal_2002} in that we work with $\Sub$-enriched categories. 
Take $\cat{V}=\Sub$ and consider a $\Sub$-Freyd category $\vfreyd{C} \colon \cat{M}\op \times \cat{M} \to \Sub$; it comes equipped with a premonoidal-like structure via $\mathsf{par}$ and $\mathsf{idt}$.
We call a morphism $f \In \vfreyd{C}(a,b)$ which is a member of the distinguished subset
a \emph{distinguished morphism}. We will show they are central in the premonoidal sense.

First observe that $\mathsf{idt} \colon \left(1, 1\right) \to \vfreyd{C}(a, a)$ is a $\Sub$ morphism, so $\mathsf{idt}(\star) \In \vfreyd{C}(a, a)$ is distinguished.
Thus, for $g \in \vfreyd{C}(a', b')$ we find $\big(\mathsf{idt}(\star), g\big) \in \vfreyd{C}\left(a, a\right) \otimes \vfreyd{C}(a', b')$ by definition of $\otimes$.
Hence the pair is in the domain of $\mathsf{par}$, giving $\mathsf{par}\big(\mathsf{idt}(\star), g\big) \in \vfreyd{C}(a \oplus a', a \oplus b')$ which we denote by $a \rtimes_{\mathsf{par}} g$.
Similarly, for any $f \in \vfreyd{C}(a, b)$ we have $f \ltimes_{\mathsf{par}} b' \in \vfreyd{C}(a \oplus b', b \oplus b')$.
We may also construct $f \ltimes_{\mathsf{par}} a'$ and $b \rtimes_{\mathsf{par}} g$.
Hence it makes sense to ask if $\mathsf{seq}(a \rtimes_{\mathsf{par}} g, f \ltimes_{\mathsf{par}} b') = \mathsf{seq}(f \ltimes_{\mathsf{par}} a', b \rtimes_{\mathsf{par}} g)$, and if this equation (and its mirrored version by placing $g$ on the left) holds for all $f$, we call $g$ \emph{central} in analogy to the binoidal case from \Cref{def:binoidal}.

Next we claim that distinguished morphisms $g \In \vfreyd{C}(a', b')$ are central.
Note that $\big(\mathsf{idt}(\star), g\big) \In \vfreyd{C}(a', a') \times \vfreyd{C}(a', b')$ and $\big(g, \mathsf{idt}(\star)\big) \In \vfreyd{C}(a', b') \times \vfreyd{C}(b', b')$ are distinguished and in the domain of $\mathsf{seq}$.
For any $f \in \vfreyd{C}(a, b)$, we have $\big((\mathsf{idt}(\star), f), (g, \mathsf{idt}(\star))\big) \in \big(\vfreyd{C}(a, a) \times \vfreyd{C}(a, b)\big) \otimes \big(\vfreyd{C}(a', b') \times \vfreyd{C}(b', b')\big)$ and similarly $\big((f, \mathsf{idt}(\star)), (\mathsf{idt}(\star)), g\big) \in \big(\vfreyd{C}(a, b) \times \vfreyd{C}(b, b)\big) \otimes \big(\vfreyd{C}(a', a') \times \vfreyd{C}(a', b')\big)$ by definition of $\otimes$ and are thus in the domain of $\mathsf{seq} \otimes \mathsf{seq}$.
We now apply $\mathsf{par} . (\mathsf{seq} \otimes \mathsf{seq})$ to each pair and find they equal $\mathsf{par}\left(f , g\right)$.
\Cref{ax:mon.mor.m.mu} states $\mathsf{par} . ( \mathsf{seq} \otimes \mathsf{seq} ) = \mathsf{seq} . ( \mathsf{par} \times \mathsf{par} ) . \zeta$ and therefore $\mathsf{seq}(a \rtimes_{\mathsf{par}} g, f \ltimes_{\mathsf{par}} b') = \mathsf{par}(f , g) = \mathsf{seq}(f \ltimes_{\mathsf{par}} a', b \rtimes_{\mathsf{par}} g)$ (and the mirrored equation analogously), so $g$ is central.

Distinguished morphisms have their centrality preserved by $\Sub$-Freyd maps as they are mapped to distinguished morphisms, but central morphisms need not be distinguished.
Thus, \Cref{def:v-freydmorphism} ensures that membership in the distinguished subset is preserved by $\Sub$-Freyd maps, so centrality of distinguished morphisms of $\vfreyd{C}$ is preserved by all maps.
Furthermore, bifunctorality of $\vfreyd{C}$ ensures that for all $f \in \cat{M}\left(a, b\right)$, $\vfreyd{C}\left(\id, f\right)\left(\mathsf{idt}\left(\star\right)\right) \In \vfreyd{C}(a, b)$, and so the image of $\cat{M}$ is central and this centrality is preserved.
The same is true for a Freyd category $J \colon \cat{M} \to \cat{C}$, the image of $\cat{M}$ under $J$ is central and this centrality is preserved by all morphisms of Freyd categories.
This preservation requirement is the difference between Freyd categories and $\Sub$-Freyd categories: the latter can require more central morphisms than the image of $\cat{M}$ to have centrality preserved.
The rest of this subsection proves that there is an adjunction between $\Freyd$ and $\Sub\text{-}\Freyd$. The left adjoint $\mathfrak{F} \colon \Freyd \to \Sub\text{-}\Freyd$ is a free functor that only requires the image of $\cat{M}$ to be preserved.
The right adjoint $\mathfrak{U} \colon \Sub\text{-}\Freyd \to \Freyd$ forgets the extra distinguished central morphisms.

\begin{proposition}\label{prop:freyd-free}
  There is a functor $\mathfrak{F} \colon \cat{Freyd} \to \Sub\text{-}\Freyd$ defined on objects as $\mathfrak{F}(\cat{C})(a,b)= \big(J(\cat{M}(a,b)), \cat{C}(a,b)\big)$ and $\mathfrak{F}(\cat{C})(f,g)=\cat{C}(Jf,Jg)$.
\end{proposition}
\begin{proof}[Proof sketch]
  $\mathfrak{F}(\cat{C})$ is well-defined on morphisms because $J$ is identity-on-objects, and it is bifunctorial by bifunctorality of hom and functorality of $J$. The structure maps are:
  \begin{itemize}
    \item $\mathsf{idt} \colon (1, 1) \to \mathfrak{F}(\cat{C})(a, a)$ is $* \mapsto \id$;
    \item $\mathsf{seq} \colon \mathfrak{F}(\cat{C})(a, b) \times \mathfrak{F}(\cat{C})(b, c) \to \mathfrak{F}(\cat{C})(a, c)$ is $(f, g) \mapsto g . f$;
    \item $\mathsf{zero} \colon (1, 1) \to \mathfrak{F}(\cat{C})(e, e)$ is $* \mapsto \id$;
    \item $\mathsf{par} \colon \mathfrak{F}(\cat{C})(a_1, b_1) \otimes \mathfrak{F}(\cat{C})(a_2, b_2) \to \mathfrak{F}(\cat{C})(a_1 \oplus a_2, b_1 \oplus b_2)$ is $(f_1,f_2) \mapsto f_1 \otimes f_2$; this is well-defined whether $(f_1, f_2)$ is in $J(\cat{M}(a_1, b_1)) \times \cat{C}(a_2, b_2)$ or is in $\cat{C}(a_1,b_1) \times J(\cat{M}(a_2,b_2))$ as $J$ preserves centrality of $\cat{M}=Z(\cat{M})$.
  \end{itemize}
  The (extra)naturality of the structure maps comes from the extranaturality of composition, functorality of $\cat{M}$'s monoidal product, and $J$ being a strict premonoidal functor preserving centrality.
  \Cref{ax:mon.mor.id,ax:mon.mor.ass} are true by $\cat{C}$'s composition, \cref{ax:obj.lax.ass,ax:obj.lax.id} follow from the strict premonoidality of $J$ and the naturality of unitors and associators, and \cref{ax:mon.mor.e.eta,ax:mon.mor.m.eta} are trivial.
  Finally, \cref{ax:mon.mor.e.mu,ax:mon.mor.m.mu} follow from $\cat{C}$'s premonoidal structure.

  Finally, it is easy to check that $\mathfrak{F}(F)=F$ is well-defined and functorial.
\end{proof}

\begin{proposition}\label{prop:freyd-forgetful}
  There is a functor $\mathfrak{U} \colon \Sub\text{-}\Freyd \to \Freyd$ that sends an object $\vfreyd{C} \colon \cat{M}\op \times \cat{M} \to \Sub$ to the functor $J \colon \cat{M} \to \mathfrak{U}(\cat{C})$ defined as follows:
  \begin{itemize}
    \item the category $\mathfrak{U}(\vfreyd{C})$ has the same objects as $\cat{M}$ but homsets $\mathfrak{U}(\vfreyd{C})(a, b) = A$ where $(X, A) \coloneqq \vfreyd{C}(a, b)$, with composition $g . f = \mathsf{seq}(f, g)$, and identity $\id{}_a = \mathsf{idt}(\star)$;
    \item the functor $J$ is the identity on objects and $J(f)=\vfreyd{C}(\id{}_a, f)(\mathsf{idt}(\star))$ on morphisms;  
    \item the binoidal structure on $\mathfrak{U}(\vfreyd{C})$ is $a \rtimes b = a \ltimes b = a \oplus_{\cat{M}} b$ on objects and $a \rtimes f = \mathsf{par}(\mathsf{idt}(\star), f)$ and $f \ltimes b = \mathsf{par}(f, \mathsf{idt}(\star))$ on morphisms.
  \end{itemize}
\end{proposition}
\begin{proof}[Proof sketch]
  It is mechanical to check that $\mathfrak{U}(\vfreyd{C})$ is a well-defined Freyd category.
  Given a morphism $F = \left(F_0, F_1\right)$ from $\vfreyd{C} \colon \cat{M}\op \times \cat{M} \to \Sub$ to $\vfreyd{C}' \colon \cat{M}'{}\op \times \cat{M}' \to \Sub$, we must define a morphism $\mathfrak{U}\left(F\right) \colon J_{\mathfrak{U}\left(\vfreyd{C}\right)} \to J_{\mathfrak{U}\left(\vfreyd{C}'\right)}$.
  We define $\mathfrak{U}\left(F\right)_0$ to be the strong monoidal functor $F_0$, and define $\mathfrak{U}\left(F\right)_1$ as $F_0$ on objects and as $F_1$ on homsets. This is a well-defined morphism of Freyd categories.
  It is straightforward to verify that $\mathfrak{U}$ is functorial.
\end{proof}

\begin{theorem}
  The functors of Propositions~\ref{prop:freyd-free} and~\ref{prop:freyd-forgetful} form an adjunction $\mathfrak{F} \dashv \mathfrak{U}$.
\end{theorem}
\begin{proof}[Proof sketch]
  For the unit $\boldsymbol{\eta}$ of the adjunction we may take the identity as a short calculation shows that $\mathfrak{U} \mathfrak{F} = \mathrm{Id}_{\Freyd}$. A second calculation shows that for a $\Sub$-Freyd category $\vfreyd{C} \colon \cat{M}\op \times \cat{M} \to \Sub$, we have $\mathfrak{F}\mathfrak{U}\left(\vfreyd{C}\right)\left(a, b\right) = \big(\vfreyd{C}(\id, \cat{M}(a, b))(\mathsf{idt}(\star)), \vfreyd{C}(a, b)\big)$, and so each component $\boldsymbol{\epsilon}_{\vfreyd{C}} \colon \mathfrak{F}\mathfrak{U}\left(\vfreyd{C}\right) \to \vfreyd{C}$ of the counit can be defined as ${\boldsymbol{\epsilon}_{\vfreyd{C}}}_0 = \mathrm{Id}_{\cat{M}}$ and ${\boldsymbol{\epsilon}_{\vfreyd{C}}}_1 = \id{}_{{\vfreyd{C}\left(a, b\right)}} \colon \mathfrak{F}\mathfrak{U}\left(\vfreyd{C}\right)\left(a, b\right) \to \vfreyd{C}\left(a, b\right)$.
  Note that the underlying $\Set$ map for ${\boldsymbol{\epsilon}_{\vfreyd{C}}}_1$ is the identity map, but this is not an identity in $\Sub$.
  This counit is natural, and this unit and counit satisfy the zig-zag identities for an adjunction.
\end{proof}

Recall that an adjunction $F \dashv G$ with unit $\eta \colon \mathrm{Id} \to G F$ and counit $\epsilon \colon F G \to \mathrm{Id}$ is \emph{idempotent} if any of $F \eta$, $\epsilon F$, $\eta G$, or $G \epsilon$ are invertible~\cite[Section~3.8]{grandis_category_2021}.
In the case of the previous theorem, clearly $\mathfrak{F} \boldsymbol{\eta}$ is invertible as $\boldsymbol{\eta}$ is the identity, so this adjunction is idempotent.
This leads to the following theorem detailing just how $\Sub\text{-}\Freyd$ generalises $\Freyd$.

\begin{theorem}\label{thm:freyd-in-subset-freyd}
  The full coreflective subcategory of $\Sub\text{-}\Freyd$ consisting of objects $\vfreyd{C} \colon \cat{M}\op \times \cat{M} \to \Sub$ for which $\vfreyd{C}\left(a, b\right)$ has the distinguished subset $\vfreyd{C}\left(\id, \cat{M}\left(a, b\right)\right)\left(\mathsf{idt}\left(\star\right)\right)$ is equivalent to $\Freyd$.
\end{theorem}
\begin{proof}[Proof sketch]
  The following is a general fact about idempotent adjunctions~\cite[Section~3.8]{grandis_category_2021}: 
  if $F \dashv G$ is an idempotent adjunction with associated monad $T = G F$ and comonad $S = F G \colon \cat{A} \to \cat{A}$, then the category of algebras of $T$ is equivalent to the category of coalgebras of $S$, and the category of coalgebras of $S$ is a full coreflective subcategory of $\cat{A}$ given by the objects of $\cat{A}$ for which $\epsilon \colon S A \to A$ is invertible.

  The category of algebras for the monad $\mathfrak{U} \mathfrak{F} = \mathrm{Id}$ is equivalent to $\Freyd$, which is therefore a full coreflective subcategory of $\Sub\text{-}\Freyd$.
  Furthermore, we can characterize the objects of this subcategory as $\Sub$-Freyd categories $\vfreyd{C}$ for which to $\boldsymbol{\epsilon} \colon \mathfrak{F}\mathfrak{U}\left(\vfreyd{C}\right) \to \vfreyd{C}$ is invertible.
  Concretely, this means ${\boldsymbol{\epsilon}_{\vfreyd{C}}}_1$ must be invertible in $\Sub$.
  But the underlying $\Set$ map is the identity, establishing the claim.
\end{proof}

\section{Abstract characterisation}\label{sec:characterisation}

Definition~\ref{def:v-freyd} is a very concrete way to specify a $\cat{V}$-Freyd category, involving a nontrivial amount of data and axioms. Yet it fits together, as we show in this subsection by giving a characterisation in the style of~\cite{jacobsheunenhasuo:arrows}.
Recall that a natural transformation between lax monoidal functors is \emph{monoidal} when it respects the coherence maps $\mu$ and $\eta$.
  Write $\cat{MonCat_{lax}}\big(\cat{C}, \cat{D}\big)$ for the category of lax monoidal functors from $\cat{C}$ to $\cat{D}$ and monoidal natural transformations between them.
%
If $\cat{A}$ and $\cat{B}$ are monoidal categories, so are $\cat{A}\op$ and $\cat{A} \times \cat{B}$, with componentwise structure.
Thus we may consider $\cat{MonCat_{lax}}\big( \cat{M}\op \times \cat{M}, \cat{V} \big)$ for the monoidal category $(\cat{V}, \mzero, \ezero)$.
We will lift the other monoidal structure $(\cat{V}, \mone, \eone)$ to $\cat{MonCat_{lax}}\big( \cat{M}\op \times \cat{M}, \cat{V} \big)$ and prove that a $\cat{V}$-Freyd category is exactly a monoid with respect to this monoidal structure, under additional assumptions on $\cat{V}$.
Most proofs are deferred to \Cref{sec:proofs:characterisation}.

\begin{definition}\label{def:cocomplete}
  A duoidal category $\cat{V}$ is a \emph{cocomplete duoidal category} if $\cat{V}$ is cocomplete and $\mzero$ and $\mone$  are cocontinuous in each argument.
  In a cocomplete duoidal category, the following diagrams and their symmetric versions commute:
  \[\begin{tikzcd}[column sep=small, row sep=small, global scale=0.8]
    {\ezero \mzero \text{colim}(D)} &[-28pt]&[-28pt] {\text{colim} \left( \ezero \mzero D \right)} \\
    & {\text{colim}(D)}
    \arrow["\simeq"', from=2-2, to=1-3]
    \arrow["\simeq", from=2-2, to=1-1]
    \arrow["\simeq"' {yshift=3pt}, from=1-3, to=1-1]
  \end{tikzcd}
  \qquad
  \begin{tikzcd}[column sep=small, row sep=small, global scale=0.8]
    {\eone \mone \text{colim}(D)} &[-28pt]&[-28pt] {\text{colim} \left( \eone \mone D \right)} \\
    & {\text{colim}(D)}
    \arrow["\simeq"', from=2-2, to=1-3]
    \arrow["\simeq", from=2-2, to=1-1]
    \arrow["\simeq"' {yshift=3pt}, from=1-3, to=1-1]
  \end{tikzcd}\]
  where the top isomorphism is colimit preservation and the others are induced by unitors.
\end{definition}

The rest of this subsection assumes that $\cat{V}$ is a cocomplete duoidal category; importantly, this is satisfied for presheaf categories.
This restriction will be mitigated in \Cref{subsec:yoneda} for small $\cat{V}$.
We also assume that $\cat{M}$ is small.
All laxness is with respect to $(\cat{V}, \mzero, \ezero)$.
We now lift $(\cat{V}, \mone, \eone)$; first the unit, then composition.

\begin{proposition}\label{prop:enriched-hom}
  There is a lax monoidal functor $\ehom_\cat{M} \colon \cat{M}\op \times \cat{M} \to \cat{V}$ defined on objects as
  $\ehom_\cat{M}(a,b) = \coprod_{\sigma \in \hom_\cat{M}(a,b)} I$.
\end{proposition}

\begin{proposition}\label{prop:profunctor-comp}
  If $S,\,T \colon \cat{M}\op \times \cat{M} \to \cat{V}$ are lax monoidal functors, the functor $S \epcomp T \colon \cat{M}\op \times \cat{M} \to \cat{V}$ defined using coends as
  $(S \epcomp T)(a, c) = \int^{b} T(a, b) \mone S(b, c)$
  is lax monoidal.
\end{proposition}


\begin{proposition}\label{prop:moncatlax-monoidal}
  $\big( \cat{MonCat_{lax}}( \cat{M}\op \times \cat{M}, \cat{V} ), \epcomp, \ehom_{\cat{M}} \big)$ is a monoidal category.
\end{proposition}
\begin{proof}
  \Cref{lem:pcomp-functorial,lem:pcomp-ehom-unit,lem:pcomp-assoc} in \Cref{sec:proofs:characterisation} show that the $\mone$-composition is functorial, associative, and has $\ehom_\cat{M}$ as left and right unit. That leaves only the triangle and pentagon identities, which follow from cocontinuity and the equivalent identities for $\mone$.
\end{proof}

With these preparations we can characterise $\cat{V}$-Freyd categories abstractly.

\begin{theorem}\label{thm:characterisation}
  Let $\cat{V}$	be a cocomplete duoidal category. Then a $\cat{V}$-Freyd category $\vfreyd{C} \colon \cat{M} \times \cat{M}\op \to \cat{V}$ is exactly a monoid in $\cat{MonCat_{lax}}( \cat{M}\op \times \cat{M}, \cat{V} )$.
\end{theorem}
\begin{proof}[Proof sketch]
  A monoid $\vfreyd{C}$ in $\cat{MonCat_{lax}}( \cat{M}\op \times \cat{M}, \cat{V} )$ consists of two maps $e \colon \ehom_{\cat{M}} \to \vfreyd{C}$ and $m \colon \vfreyd{C} \epcomp \vfreyd{C} \to \vfreyd{C}$, inducing $\mathsf{idt}$ and $\mathsf{seq}$ satisfying unit and associativity conditions.
  The lax monoidal structure of $\vfreyd{C}$ gives $\mathsf{zero}$ and $\mathsf{par}$ respectively, so identity and associativity conditions follow.
  Finally, the components of $e$ and $m$ are monoidal natural transformations, ensuring that $\mathsf{idt}$ and $\mathsf{seq}$ respect $\mathsf{zero}$ and $\mathsf{par}$.
\end{proof}
We note that by Fujii's observations \cite{fujii:unified}, PROs and PROPs are equivalent to $\Set$-Freyd categories over $\cat{N}$ and $\cat{P}$ respectively because $(\Set, \times, \times)$ is a cocomplete duoidal category.

\section{Change of enrichment}\label{sec:changeofenrichment}

After defining enriched categories, a natural next step is to consider a change of enrichment. Any monoidal functor $\cat{V} \to \cat{W}$ induces a functor $\cat{V}\text{-}\cat{Cat} \to \cat{W}\text{-}\cat{Cat}$. We will show that the same holds for the appropriate type of functors between duoidal categories and enriched Freyd categories (in \Cref{subsec:lifting}). We will then use that to alleviate the restriction of duoidal cocompleteness on the abstract characterisation of \Cref{sec:characterisation} (in \Cref{subsec:yoneda}) at the cost of losing a direction of the correspondence. Finally, changing enrichment along a forgetful functor gives an underlying (unenriched) Freyd category $J \colon \cat{M} \to \cat{C}$ with $\cat{C}$ monoidal, which we show recovers the pure computations in the examples of \Cref{sec:examples} (in \Cref{subsec:forgetful}).

\subsection{Lifting duoidal functors}\label{subsec:lifting}

To talk about change of enrichment, we first need to define the appropriate type of functor between the enriching categories along which to change. 

\begin{definition}\label{def:double-lax}\cite[Definition~6.54]{aguiar2010monoidal}
  Take duoidal categories $\left(\cat{V}, \mzero_\cat{V}, \ezero_\cat{V}, \mone_\cat{V}, \eone_\cat{V}\right)$ and $\left(\cat{W}, \mzero_\cat{W}, \ezero_\cat{W}, \mone_\cat{W}, \eone_\cat{W}\right)$.
  A functor $F \colon \cat{V} \!\to\! \cat{W}$ is a \emph{double lax monoidal functor} when equipped with $\eta_\mzero$, $\mu_\mzero$, $\eta_\mone$, and $\mu_\mone$ such that $\left(F, \eta_\mzero, \mu_\mzero\right)$ is lax monoidal for $\mzero_\cat{V}$ and $\mzero_\cat{W}$, $\left(F, \eta_\mone, \mu_\mone\right)$ is lax monoidal for $\mone_\cat{V}$ and $\mone_\cat{W}$, and the following diagrams commute:
  \[\begin{tikzcd}[column sep=tiny, global scale=0.8]
	{\left( F(A) \mone_{\cat{W}} F(B) \right) \mzero_{\cat{W}} \left( F(C) \mone_{\cat{W}} F(D)\right)} & {\left( F(A) \mzero_{\cat{W}} F(C) \right) \mone_{\cat{W}} \left( F(B) \mzero_{\cat{W}} F(D)\right)} \\
	{F\left( A \mone_{\cat{V}} B \right) \mzero_{\cat{W}} F\left( C \mone_{\cat{V}} D\right)} & {F\left( A \mzero_{\cat{V}} C \right) \mone_{\cat{W}} F\left( B \mzero_{\cat{V}} D\right)} \\
	{F\left(\left( A \mone_{\cat{V}} B \right) \mzero_{\cat{V}} \left( C \mone_{\cat{V}} D\right)\right)} & {F\left(\left( A \mzero_{\cat{V}} C \right) \mone_{\cat{V}} \left( B \mzero_{\cat{V}} D\right)\right)}
	\arrow["{\mu_{\mone} \mzero \mu_{\mone}}"', from=1-1, to=2-1]
	\arrow["{\mu_{\mzero}}"', from=2-1, to=3-1]
	\arrow["{\zeta}" {yshift=3pt}, from=1-1, to=1-2]
	\arrow["{\mu_{\mzero} \mone \mu_{\mzero}}", from=1-2, to=2-2]
	\arrow["{\mu_{\mone}}", from=2-2, to=3-2]
	\arrow["{F\zeta}"', from=3-1, to=3-2]
  \end{tikzcd}\hspace{-9pt}
  \begin{tikzcd}[column sep=tiny, global scale=0.8]
	{F(\ezero_{\cat{V}})} & {F(\eone_{\cat{V}})} \\
	{\ezero_{\cat{W}}} & {\eone_{\cat{W}}}
	\arrow["F\epsilon" {yshift=3pt}, from=1-1, to=1-2]
	\arrow["{\eta_{\mzero}}", from=2-1, to=1-1]
	\arrow["\epsilon"', from=2-1, to=2-2]
	\arrow["{\eta_{\mone}}"', from=2-2, to=1-2]
 \end{tikzcd}\]
 \[\begin{tikzcd}[column sep=small, global scale=0.8]
	{\ezero_{\cat{W}}} & {F(\ezero_{\cat{V}})} & {F\left( \ezero_{\cat{V}} \mone_{\cat{V}} \ezero_{\cat{V}} \right)} \\
	{\ezero_{\cat{W}} \mone_{\cat{W}} \ezero_{\cat{W}}} && {F(\ezero_{\cat{V}}) \mone_{\cat{W}} F(\ezero_{\cat{V}})}
	\arrow["{\eta_{\mzero}}", from=1-1, to=1-2]
	\arrow["F\Delta", from=1-2, to=1-3]
	\arrow["\Delta"', from=1-1, to=2-1]
	\arrow["{\eta_{\mzero} \mone \eta_{\mzero}}"', from=2-1, to=2-3]
	\arrow["{\mu_{\mone}}"', from=2-3, to=1-3]
 \end{tikzcd}
 \begin{tikzcd}[column sep=small, global scale=0.8]
	{\eone_{\cat{W}}} & {F(\eone_{\cat{V}})} & {F\left( \eone_{\cat{V}} \mzero_{\cat{V}} \eone_{\cat{V}} \right)} \\
	{\eone_{\cat{W}} \mzero_{\cat{W}} \eone_{\cat{W}}} && {F(\eone_{\cat{V}}) \mzero_{\cat{W}} F(\eone_{\cat{V}})}
	\arrow["{\eta_{\mone}}", from=1-1, to=1-2]
	\arrow["F\nabla"', from=1-3, to=1-2]
	\arrow["\nabla", from=2-1, to=1-1]
	\arrow["{\eta_{\mone} \mzero \eta_{\mone}}"', from=2-1, to=2-3]
	\arrow["{\mu_{\mzero}}"', from=2-3, to=1-3]
  \end{tikzcd}\]
\end{definition}

Here now is the change-of-enrichment theorem for duoidally enriched Freyd categories.

\begin{theorem}\label{thm:double-lax-change}
  Let $F \colon \cat{V} \to \cat{W}$ be a double lax monoidal functor.
  For a $\cat{V}$-Freyd category $\vfreyd{C} \colon \cat{M}\op \times \cat{M} \to \cat{V}$, define $\overline{F}(\vfreyd{C})(a, b) \coloneqq F(\vfreyd{C}(a, b))$ with structure maps $\mathsf{idt}_F \coloneqq F \mathsf{idt} . \eta_{\mone}$, $\mathsf{seq}_F \coloneqq F \mathsf{seq} . \mu_{\mone}$, $\mathsf{zero}_F \coloneqq F \mathsf{zero} . \eta_{\mzero}$, and $\mathsf{par}_F \coloneqq F \mathsf{par} . \mu_{\mzero}$.
  For a map $G = (G_0, G_1) \colon \vfreyd{C} \to \vfreyd{C}'$, define $\overline{F}(G) \coloneqq (G_0, F G_1)$.
  This $\overline{F}$ is a functor $\cat{V}\text{-}\Freyd \to \cat{W}\text{-}\Freyd$.
\end{theorem}
\begin{proof}
  See \Cref{prf:double-lax-change}.
\end{proof}

\begin{example}\label{ex:label-change-enrich}
  Let $M$ and $N$ be separated monoids and $\phi \colon M \to N$ a homomorphism such that $\phi(m) \mathop{\Vert} \phi(m')$ implies $m \mathop{\Vert} m'$.
  Then $\phi$ induces a double lax monoidal functor $\phi_* \colon \Lab_{M} \to \Lab_{N}$ given by $\ell \mapsto \phi . \ell$ on objects and $f \mapsto f$ on morphisms.
  The maps $\eta_\mzero$, $\mu_\mzero$, and $\eta_\mone$ are all identities, while $\mu_\mone \colon \{(a,a') \mid \phi.\ell(a) \Vert \phi.\ell'(a') \} \to \{(a,a') \mid \ell(a) \Vert \ell'(a') \}$ is the inclusion, and so $\phi_*$ is clearly double lax monoidal.
  Apply \Cref{thm:double-lax-change} to the example from \Cref{sub:state-monoid} along the map $\mathcal{P}_f(!) \colon \mathcal{P}_f(R) \to \mathcal{P}_f(1)$, which is a homomorphism such that $\mathcal{P}_f(!)(P) \cap \mathcal{P}_f(!)(Q) = \emptyset$ implies $P \cap Q = \emptyset$.
  We get $\mathcal{P}_f(!)_*(\vfreyd{C})(a, b) = \sum_{Q \in \mathcal{P}_f(R)} \left(\Set(a \times \Pi_Q, b \times \Pi_Q )\right) \to \mathcal{P}_f(1)$, $(Q, f) \mapsto \emptyset \text{ if } Q = \emptyset \text{, else }1$.
  This change of enrichment alters the example to only allowing maps to be put in parallel if at least one of them requires no resources.
\end{example}

\begin{example}\label{ex:state-change-enrich}
  We can use change of enrichment for the indexed state example of \Cref{subsec:indexed-state}.
  Consider \Cref{ex:duoidal-profunctor} for $(\Set, \times, 1)$ (using universes for this example to avoid size issues). There, the definition of Day convolution $\times\Day$ simplifies to $(P \times\Day Q)(a, b) = \int^{b_2, b_2} \Set(b_1 \times b_2, b) \times P(a, b_1) \times Q(a, b_2)$ and its unit becomes $k(a, b) = b$.
  The Kleisli construction turns a finitary endofunctor on $\Set$ into a profunctor as follows. 
  Define $\Kl \colon \left[\Set, \Set\right]_f \to \Prof(\Set)$ by $\Kl(F)(a, b) = \Set(a, Fb)$, and coherence maps:
  \begin{align*}
    \eta_{\mzero} \colon k &\to \Kl(\mathrm{Id})
    & \mu_{\mzero} \colon \Kl(F_1) \times\Day \Kl(F_2) &\to \Kl(F_1 \times\Day F_2) \\
    b &\mapsto \mathsf{cst}_b
    & (k, f_1, f_2) &\mapsto \lambda a. (k, f_1(a), f_2(a)) \\[10pt]
    \eta_{\mone} \colon \hom &\to \Kl(\mathrm{Id})
    & \mu_{\mone} \colon \Kl(F) \diamond \Kl(G) &\to \Kl(F \circ G) \\
    f &\mapsto f
    & (f, g) &\mapsto F g . f
  \end{align*}
  This makes $\Kl$ a double lax monoidal functor.
  \Cref{thm:double-lax-change} then gives a $\Prof(\Set)$-Freyd category defined by $\Kl(\vfreyd{C})(a, b)(x, y) \coloneqq \Set(x, \left(b \times y\right)^a)$.
\end{example}

\subsection{Yoneda embedding}\label{subsec:yoneda}

The Yoneda embedding of a small monoidal category is a strong monoidal functor with respect to Day convolution.
This extends to small duoidal categories.

\begin{proposition}\label{prop:yoneda}
  The Yoneda embedding $\cat{V} \!\to\! [\cat{V}\op, \Set]$ is a double lax monoidal functor from small $\left(\cat{V}, \mzero, \ezero, \mone, \eone\right)$ to $\big([\cat{V}\op, \Set], \mzero\Day, \cat{V}(-, \ezero), \mone\Day, \cat{V}(-, \eone)\big)$.
\end{proposition}
\begin{proof}
  See~\cite{imkelly:convolution} for the fact that it is lax monoidal for each monoidal structure separately. The diagrams of \Cref{def:double-lax} are verified straightforwardly.
\end{proof}

It follows from \Cref{thm:double-lax-change} that every $\cat{V}$-Freyd category for small $\cat{V}$ induces a $[\cat{V}\op, \Set]$-Freyd category. But $[\cat{V}\op, \Set]$ is duoidally cocomplete, so the setting in which the abstract characterisation of \Cref{thm:characterisation} applies. We conclude that the characterisation extends beyond the duoidally cocomplete setting in the sense that every $\cat{V}$-Freyd category for small $\cat{V}$ induces a monoid in $\cat{MonCat_{lax}}( \cat{M}\op \times \cat{M}, [\cat{V}\op, \Set])$.

\subsection{Forgetful functors}\label{subsec:forgetful}

Any category enriched in a monoidal category $\cat{V}$ has an underlying (unenriched) category, got by changing the enrichment along the `forgetful' monoidal functor $\cat{V}(I,-) \colon \cat{V} \to \cat{Set}$. A similar process plays out for duoidal categories.

\begin{proposition}\label{prop:forgetful}
  Let $\left(\cat{V}, \mzero, \ezero, \mone, \eone\right)$ be a duoidal category and write $\phi \colon \ezero \to \ezero \mzero \ezero$ for the inverse of the unitors.
  Then $\cat{V}(\ezero, -) \colon \cat{V} \to \Set$ is a double lax monoidal functor with coherence maps:
  \begin{align*}
    \eta_{\mzero} \colon 1 &\to \cat{V}(\ezero, \ezero)
    & \mu_{\mzero} \colon \cat{V}(\ezero, A_1) \times \cat{V}(\ezero, A_2) &\to \cat{V}(\ezero, A_1 \mzero A_2) \\
    \star &\mapsto \id
    & (f_1, f_2) &\mapsto (f_1 \mzero f_2) . \phi \\[10pt]
    \eta_{\mone} \colon 1 &\to \cat{V}(\ezero, \eone)
    & \mu_{\mone} \colon \cat{V}(\ezero, A_1) \times \cat{V}(\ezero, A_2) &\to \cat{V}(\ezero, A_1 \mone A_2) \\
    \star &\mapsto \epsilon
    & (f_1, f_2) &\mapsto (f_1 \mone f_2) . \Delta
  \end{align*}
\end{proposition}

Applying \Cref{thm:double-lax-change} along the forgetful functor of the previous proposition in the case of the examples of \Cref{sec:examples} will show that this recovers the underlying pure computations.
Note that a $\Set$-Freyd category $\vfreyd{C}$ has a trivial instance of the exchange axiom, \cref{ax:mon.mor.m.mu}, and so $\vfreyd{C}$ is a monoidal category with identity-on-objects monoidal functor $J \colon \cat{M} \to \cat{C}$.

\begin{example}\label{ex:forgetful-state-monoid}
  Applying the forgetful functor to the stateful function example of \Cref{sub:state-monoid} results in the (unenriched) category with $\Lab_{\mathcal{P}_f(R)}(\mathsf{cst}_\emptyset, \vfreyd{C}(a, b))$ as the homsets.
  Because labels are preserved, the morphisms in this (unenriched) category are exactly the elements of $\vfreyd{C}(a, b)$ which have label $\emptyset$, \textit{i.e.}\ maps $a \times 1 \to b \times 1$ which are pure functions.
\end{example}

\begin{example}\label{ex:forgetful-indexed-state}
  Changing the enrichment of the indexed state example from \Cref{subsec:indexed-state} along the forgetful functor gives the (unenriched) category with homsets $\left[\Set, \Set\right]_f(\mathrm{Id}, \vfreyd{C}(a, b))$.
  If $\phi \colon \mathrm{Id} \to (b \times ( -))^{a}$ is such a natural transformation, then the function $\phi_1 \colon 1 \to (b \times 1)^{a}$, which is equivalent to choosing a function $f \colon a \to b$, completely determines $\phi$, because for any set $X$ and $x \in X$ by naturality $1 \xrightarrow{x} X \xrightarrow{\phi_X} (b \times X)^{a} = 1 \xrightarrow{\phi_1} (b \times 1)^{a} \xrightarrow{(\id \times x).-} (b \times X)^{a}$, whence $\phi_X(x)(a) = (f(a), x)$.
  Therefore the morphisms in this (unenriched) category are all functions $a \to b$.
\end{example}

\begin{example}\label{ex:forgetful-kleislilawvere}
  Changing the enrichment of the Kleisli categories of Lawvere theories example from \Cref{sub:kleislilawvere} along the forgetful functor gives the (unenriched) category with homsets $[\Law, \Set](\mathbb{1}, \vfreyd{C}(a, b))$.
  Consider such a natural transformation $\phi \colon \mathbb{1} \to  T(-)(b)^a$.
  It is completely determined by its component at $\mathcal{S}$.
  For any $\mathcal{L}$ let $\iota \colon \mathcal{S} \to \mathcal{L}$ be the unique map, then naturality implies $\phi_{\mathcal{L}} = T(\iota)\phi_{\mathcal{S}}$.
  Furthermore, $\phi_\mathcal{S}(\star) \in T(\mathcal{S})(b)^a = b^a$. 
  So the morphisms in this (unenriched) category again are all functions $a \to b$.
\end{example}

\section{Conclusion}\label{sec:conclusion}

We have defined a version of Freyd categories enriched over any duoidal category $\cat{V}$, and morphisms between them.
We used various duoidal categories to give examples based on separation of resources, parameterised monads, and the Kleisli construction for Lawvere theories.
By enriching with $\Sub$, we have proven that the category of Freyd categories $\Freyd$ is a full coreflective subcategory of $\Sub\text{-}\Freyd$, thus establishing that $\cat{V}$-Freyd categories indeed generalise Freyd categories.
Additionally, we proved an abstract characterisation of $\cat{V}$-Freyd categories over small $\cat{M}$  for duoidally cocomplete $\cat{V}$, they are monoids in $\cat{MonCat_{lax}}\big( \cat{M}\op \times \cat{M}, \cat{V} \big)$.
Finally, we provided change of enrichment and examples thereof.

\paragraph{Future work}

There are several directions for further investigation:
\begin{itemize}
\item  
The abstract characterisation of \Cref{sec:characterisation} may be part of a larger structure, namely a \emph{bicategory with proarrow equipment}, whose objects are monoidal categories, arrows are strong monoidal functors, proarrows are lax monoidal profunctors, and cells are lax monoidal natural transformations.
In this setting, a $\cat{V}$-Freyd category would be a monad and the vertical monad morphisms would be a $\cat{V}$-Freyd morphism.
This would enable applying general constructions for monads in a bicategory.

\item
Relatedly, an \emph{fc-multicategory} structure on $\cat{MonCat_{lax}}(\cat{M}\op\times\cat{M},\cat{V})$ may bypass cocompleteness in characterising $\cat{V}$-Freyd categories as monoids.

\item
The abstract characterisation of \Cref{sec:characterisation} also uses the free $\cat{V}$-category on $\cat{M}$.
It may be fruitful to change the definition of a $\cat{V}$-Freyd category to be a $\cat{V}$-functor $J \colon \cat{M} \to \cat{C}$ where we extend $\cat{V}$-categories in a way similar to Morrison and Penneys~\cite{morrisonpenneys:braided}.

\item
Freyd categories can have the property of being \emph{closed}. In this case they induce a strong monad.
A similar definition may be possible for $\cat{V}$-Freyd categories. This could determine a higher-order semantics for effectful programs based on duoidal categories.
A nontrivial definition of closure may require a $\cat{V}$-category $\cat{M}$ that is not free.

\item
Our original motivation stemmed from the desire for semantics combining differentiable and probabilistic programming, in particular, the possibility of having a linear structure for the probabilistic fragment and a cartesian one for differentiable terms.
$\Prof$-Freyd categories may provide a useful separation to aid the desired distinction between linear and cartesian properties.
\end{itemize}

\subsubsection{Acknowledgments}
We would like to thank Robin Kaarsgaard, Ohad Kammar, and Matthew Di Meglio for their input and encouragement, as well as the reviewers of all versions of this work.

\bibliography{bibliography}

\appendix

\section{Definition of $\cat{V}$-Freyd category}\label{sec:diagrams:vfreyd}

This appendix spells out the type diagrams of \Cref{def:v-freyd} of $\cat{V}$-Freyd categories.

\begin{description}
  \item[Extranaturality of $\mathsf{idt}$:]
\[\begin{tikzcd}[column sep=tiny, global scale=0.8]
  \eone & {\vfreyd{C}(b, b)} \\
  {\vfreyd{C}(a, a)} & {\vfreyd{C}(a,b)}
  \arrow["{\mathsf{idt}}", from=1-1, to=1-2]
  \arrow["{\vfreyd{C}(f, \id)}", from=1-2, to=2-2]
  \arrow["{\mathsf{idt}}"', from=1-1, to=2-1]
  \arrow["{\vfreyd{C}(\id, f)}"' {yshift=-3pt}, from=2-1, to=2-2]
\end{tikzcd}\]
  
  \item[Extranaturality of $\mathsf{seq}$:]
\[\begin{tikzcd}[column sep=tiny, global scale=0.8]
	{\vfreyd{C}(a, b) \mone \vfreyd{C}(b', c)} & {\vfreyd{C}(a, b') \mone \vfreyd{C}(b', c)} \\
	{\vfreyd{C}(a, b) \mone \vfreyd{C}(b, c)} & {\vfreyd{C}(a, c)}
	\arrow["{\vfreyd{C}(\id, f) \mone \id}" {yshift=3pt}, from=1-1, to=1-2]
	\arrow["{\mathsf{seq}}", from=1-2, to=2-2]
	\arrow["{\id \mone \vfreyd{C}(f, \id)}"', from=1-1, to=2-1]
	\arrow["{\mathsf{seq}}"', from=2-1, to=2-2]
\end{tikzcd}\]

  \item[$\mathsf{idt}$ is the identity for $\mathsf{seq}$:]
\[\begin{tikzcd}[column sep=tiny, global scale=0.8]
	{\eone \mone \vfreyd{C}(a, b)} & {\vfreyd{C}(a, a) \mone \vfreyd{C}(a, b)} \\
	& {\vfreyd{C}(a, b)}
	\arrow["\lambda"', from=1-1, to=2-2]
	\arrow["{\mathsf{seq}}", from=1-2, to=2-2]
	\arrow["{\mathsf{idt} \mone \id}" {yshift=3pt}, from=1-1, to=1-2]
\end{tikzcd}
\qquad
\begin{tikzcd}[column sep=tiny, global scale=0.8]
	{\vfreyd{C}(a, b) \mone \eone} & {\vfreyd{C}(a, b) \mone \vfreyd{C}(b, b)} \\
	& {\vfreyd{C}(a, b)}
	\arrow["\rho"', from=1-1, to=2-2]
	\arrow["{\mathsf{seq}}", from=1-2, to=2-2]
	\arrow["{\id \mone \mathsf{idt}}" {yshift=3pt}, from=1-1, to=1-2]
\end{tikzcd}\]

  \item[$\mathsf{seq}$ is associative:]
\[\begin{tikzcd}[column sep=tiny, global scale=0.8]
	{\left( \vfreyd{C}(a, b) \mone \vfreyd{C}(b, c) \right) \mone \vfreyd{C}(c, d)} &[-18pt]&[-18pt] {\vfreyd{C}(a, b) \mone \left( \vfreyd{C}(b, c) \mone \vfreyd{C}(c, d) \right)} \\
	{\vfreyd{C}(a, c) \mone \vfreyd{C}(c, d)} && {\vfreyd{C}(a, b) \mone \vfreyd{C}(b, d)} \\
	& {\vfreyd{C}(a, d)}
	\arrow["\alpha" {yshift=3pt}, from=1-1, to=1-3]
	\arrow["{\id \mone \mathsf{seq}}", from=1-3, to=2-3]
	\arrow["{\mathsf{seq}}", from=2-3, to=3-2]
	\arrow["{\mathsf{seq} \mone \id}"', from=1-1, to=2-1]
	\arrow["{\mathsf{seq}}"', from=2-1, to=3-2]
\end{tikzcd}\]

 \item[$\mathsf{zero}$ is the identity for $\mathsf{par}$:]
\[\begin{tikzcd}[column sep=tiny, global scale=0.8]
	{\ezero \mzero \vfreyd{C}(a, b)} & {\vfreyd{C}(e, e) \mzero \vfreyd{C}(a, b)} \\
	{\vfreyd{C}(a, b)} & {\vfreyd{C}(e \oplus a, e \oplus b)}
	\arrow["\lambda"', from=1-1, to=2-1]
	\arrow["{\mathsf{zero} \mzero \id}" {yshift=3pt}, from=1-1, to=1-2]
	\arrow["{\mathsf{par}}", from=1-2, to=2-2]
	\arrow["{\vfreyd{C}(\lambda^{-1}, \lambda)}" {yshift=-3pt}, from=2-2, to=2-1]
\end{tikzcd}
\qquad
\begin{tikzcd}[column sep=tiny, global scale=0.8]
	{\vfreyd{C}(a, b) \mzero \ezero} & {\vfreyd{C}(a, b) \mzero \vfreyd{C}(e, e)} \\
	{\vfreyd{C}(a, b)} & {\vfreyd{C}(b \oplus e, b \oplus e)}
	\arrow["\rho"', from=1-1, to=2-1]
	\arrow["{\id \mzero \mathsf{zero}}" {yshift=3pt}, from=1-1, to=1-2]
	\arrow["{\mathsf{par}}", from=1-2, to=2-2]
	\arrow["{\vfreyd{C}(\rho^{-1}, \rho)}" {yshift=-3pt}, from=2-2, to=2-1]
\end{tikzcd}\]

  \item[$\mathsf{par}$ is associative:]
\[\begin{tikzcd}[column sep=tiny, global scale=0.8]
	{\left( \vfreyd{C}(a_1, b_1) \mzero \vfreyd{C}(a_2, b_2) \right) \mzero \vfreyd{C}(a_3, b_3)} & {\vfreyd{C}(a_1, b_1) \mzero \left( \vfreyd{C}(a_2, b_2) \mzero \vfreyd{C}(a_3, b_3) \right)} \\
	{\vfreyd{C}(a_1 \oplus a_2, b_1 \oplus b_2) \mzero \vfreyd{C}(a_3, b_3)} & {\vfreyd{C}(a_1, b_1) \mzero \vfreyd{C}(a_2 \oplus a_3, b_2 \oplus b_3)} \\
	{\vfreyd{C}((a_1 \oplus a_2) \oplus a_3, (b_1 \oplus b_2) \oplus b_3)} & {\vfreyd{C}(a_1 \oplus (a_2 \oplus a_3), b_1 \oplus (b_2 \oplus b_3))}
	\arrow["\alpha" {yshift=3pt}, from=1-1, to=1-2]
	\arrow["{\id \mzero \mathsf{par}}", from=1-2, to=2-2]
	\arrow["{\mathsf{par}}", from=2-2, to=3-2]
	\arrow["{\mathsf{par} \mzero \id}"', from=1-1, to=2-1]
	\arrow["{\mathsf{par}}"', from=2-1, to=3-1]
	\arrow["{\vfreyd{C}(\alpha^{-1}, \alpha)}"' {yshift=-3pt}, from=3-1, to=3-2]
\end{tikzcd}\]

  \item[$\mathsf{idt}$ respects $\mathsf{zero}$:]
\[\begin{tikzcd}[column sep=tiny, global scale=0.8]
	\ezero & \eone \\
	& {\vfreyd{C}(e, e)}
	\arrow["\epsilon", from=1-1, to=1-2]
	\arrow["{\mathsf{idt}}", from=1-2, to=2-2]
	\arrow["{\mathsf{zero}}"', from=1-1, to=2-2]
\end{tikzcd}\]

  \item[$\mathsf{idt}$ respects $\mathsf{par}$:]
\[\begin{tikzcd}[column sep=tiny, global scale=0.8]
	{\eone \mzero \eone} & {\vfreyd{C}(a, a) \mzero \vfreyd{C}(b, b)} \\
	\eone & {\vfreyd{C}(a \oplus b, a \oplus b)}
	\arrow["{\mathsf{par}}", from=1-2, to=2-2]
	\arrow["{\mathsf{idt} \mzero \mathsf{idt}}" {yshift=3pt}, from=1-1, to=1-2]
	\arrow["\nabla"', from=1-1, to=2-1]
	\arrow["{\mathsf{idt}}"', from=2-1, to=2-2]
\end{tikzcd}\]

  \item[$\mathsf{seq}$ respects $\mathsf{zero}$:]
\[\begin{tikzcd}[column sep=tiny, global scale=0.8]
	\ezero & {\ezero \mone \ezero} \\
	{\vfreyd{C}(e, e)} & {\vfreyd{C}(e, e) \mone \vfreyd{C}(e, e)}
	\arrow["{\mathsf{zero}}"', from=1-1, to=2-1]
	\arrow["{\mathsf{zero} \mone \mathsf{zero}}", from=1-2, to=2-2]
	\arrow["\Delta", from=1-1, to=1-2]
	\arrow["{\mathsf{seq}}" {yshift=-3pt}, from=2-2, to=2-1]
\end{tikzcd}\]

  \item[$\mathsf{seq}$ respects $\mathsf{par}$:]
\[\begin{tikzcd}[column sep=tiny, global scale=0.8]
	{\left( \vfreyd{C}(a_1, b_1) \!\mone\! \vfreyd{C}(b_1, c_1) \right) \!\mzero\! \left( \vfreyd{C}(a_2, b_2) \!\mone\! \vfreyd{C}(b_2, c_2) \right)} &[-36pt]&[-36pt] {\left( \vfreyd{C}(a_1, b_1) \!\mzero\! \vfreyd{C}(a_2, b_2) \right) \!\mone\! \left( \vfreyd{C}(b_1, c_1) \!\mzero\! \vfreyd{C}(b_2, c_2) \right)} \\
	{\vfreyd{C}(a_1, c_1) \mzero \vfreyd{C}(a_2, c_2)} && {\vfreyd{C}(a_1 \oplus a_2, b_1 \oplus b_2) \mone \vfreyd{C}(b_1 \oplus b_2, c_1 \oplus c_2)} \\
	& {\vfreyd{C}(a_1 \oplus a_2, c_1 \oplus c_2)}
	\arrow["\zeta", from=1-1, to=1-3]
	\arrow["{\mathsf{par} \mone \mathsf{par}}", from=1-3, to=2-3]
	\arrow["{\mathsf{seq}}", from=2-3, to=3-2]
	\arrow["{\mathsf{seq} \mzero \mathsf{seq}}"', from=1-1, to=2-1]
	\arrow["{\mathsf{par}}"', from=2-1, to=3-2]
\end{tikzcd}\]
\end{description}

\section{Proofs for abstract characterisation}\label{sec:proofs:characterisation}

This appendix contains proofs of the abstract characterisation of $\cat{V}$-Freyd categories of \Cref{sec:characterisation}. They rely on properties of $\cat{V}$-Freyd categories listed in the following four lemmas, that are mechanical to verify.

%
%
%

\begin{lemma}\label{lem:conc-str-mor-unr}
  The unitors of $\mone$ respect $\mathsf{zero}$ and $\mathsf{par}$:
  \begin{align*}
    \rho . \mathsf{zero} &= ( \mathsf{zero} \mone \epsilon ) . \Delta&
    \mathsf{zero} . \lambda &= ( \epsilon \mone \mathsf{zero} ) . \Delta \\ 
    \rho . \mathsf{par} &= ( \mathsf{par} \mone \nabla ) . \zeta . ( \rho \mzero \rho ) &
    \mathsf{par} . \lambda &= ( \mathsf{par} \mone \nabla ) . \zeta . ( \lambda \mzero \lambda )
  \end{align*}
\end{lemma}

\begin{lemma}\label{lem:conc-str-mor-assr}
  The associator of $\mone$ respects $\mathsf{zero}$ and $\mathsf{par}$:
  \begin{align*}
    \alpha . ( \mathsf{zero} \mone ( \mathsf{zero} \mone \mathsf{zero} ) ) . ( \id \mone \Delta ) . \Delta &= ( ( \mathsf{zero} \mone \mathsf{zero} ) \mone \mathsf{zero} ) . ( \Delta \mone \id ) . \Delta \\
    \alpha . ( \mathsf{par} \mone ( \mathsf{par} \mone \mathsf{par} ) ) . ( \id \mone \zeta ) . \zeta &= ( ( \mathsf{par} \mone \mathsf{par} ) \mone \mathsf{par} ) ) . ( \zeta \mone \id ) . \zeta . ( \alpha \mzero \alpha )
  \end{align*}
\end{lemma}

\begin{lemma}\label{lem:conc-str-comp-id}
  The unitors of $\mzero$ respect $\mathsf{zero}$ and $\mathsf{par}$:
  \begin{align*}
   \id &= ( \mathsf{par} \mone \mathsf{par} ) . \zeta . ( \id \mzero ( ( \mathsf{zero} \mone \mathsf{zero} ) . \Delta ) ) . \rho 
   \\
   \id &= ( \mathsf{par} \mone \mathsf{par} ) . \zeta . ( ( ( \mathsf{zero} \mone \mathsf{zero} ) . \Delta ) \mzero \id ) . \lambda
  \end{align*}
\end{lemma}

\begin{lemma}\label{lem:conc-str-comp-ass}
  The associator of $\mzero$ respects $\mathsf{par}$:
  \begin{align*}
    ( ( \mathsf{par} . ( \mathsf{par} \mzero \id ) ) \mone ( \mathsf{par} . ( \mathsf{par} \mzero \id ) ) . \zeta . ( \zeta \mzero \id ) = \\
    ( ( \mathsf{par} . ( \id \mzero \mathsf{par} ) ) \mone ( \mathsf{par} . ( \id \mzero \mathsf{par} ) ) . \zeta . ( \id \mzero \zeta ) . \alpha
  \end{align*}
\end{lemma}

The previous lemmas require all the axioms of a duoidal category between them, except for $\mone$ being a monoid in $(\cat{V}, \mzero, \ezero)$. This latter property is used in the abstract characterisation.

\begin{proof}[Proof of \Cref{prop:enriched-hom}]
  Bifunctorality is inherited from $\hom_\cat{M}$. The coherence morphisms making it lax monoidal are $\eta \colon \ezero \xrightarrow{\epsilon} \eone \xrightarrow{\iota_{\id[0]}} \coprod\nolimits_{\sigma} \eone \cong \ehom_\cat{M}(e, e)$ and
  \begin{align*}
    \mu \colon& 
      \big(\textstyle\coprod\nolimits_{\sigma_1} \eone \big) \mzero \big(\textstyle\coprod\nolimits_{\sigma_2} \eone \big) 
      \cong \textstyle\coprod\nolimits_{\sigma_1, \sigma_2} \eone \mzero \eone
      \xrightarrow{\textstyle\coprod \nabla} \textstyle\coprod\nolimits_{\sigma_1, \sigma_2} I
      \xrightarrow{[\iota_{\sigma_1 \oplus \sigma_2}]_{\sigma_1, \sigma_2}} \textstyle\coprod\nolimits_{\sigma}\text.
  \end{align*}
  The coherence diagrams commute by cocontinuity and the monoidal structure $(\eone, \nabla, \epsilon)$.
\end{proof}

\begin{proof}[Proof of \Cref{prop:profunctor-comp}]
  The coherence morphisms are:
  \begin{align*}
    \eta_{S \epcomp T} \colon& \ezero \xrightarrow{\Delta} \ezero \mone \ezero \xrightarrow{\eta_S \mone \eta_T} T(e, e) \mone S(e, e) \to \textstyle\int^{b} T(e, b) \mone S(b, e) \cong (S \epcomp T)(e, e)
    \\
    \mu_{S \epcomp T} \colon& (S \epcomp T)(a, c) 
    \mzero (S \epcomp T)(a', c') \\
    \simeq& \textstyle\int^{b, b'} \left(T(a, b)\!\mone\!S(b, c)\right) \mzero \left(T(a', b')\!\mone\!S(b', c')\right) \\
    \xrightarrow{\int\!\zeta}& \textstyle\int^{b, b'} \left(T(a, b)\!\mzero\!T(a', b')\right) \mone \left(S(b, c)\!\mzero\!S(b', c')\right) \\
    \xrightarrow{\int\!\mu_T \mone \mu_S}& \textstyle\int^{b, b'} T(a \oplus a', b \oplus b') \mone S(b \oplus b', c \oplus c') \\
    \rightarrow& \textstyle\int^{b} T(a \oplus a', b) \mone S(b, c' \oplus c') 
    \simeq (S \epcomp T)(a \oplus a', c \oplus c')
  \end{align*}
  Cocontinuity and \Cref{lem:conc-str-comp-id,lem:conc-str-comp-ass} finish the proof.
\end{proof}

\begin{lemma}\label{lem:pcomp-functorial}
  The $\mone$-composition of \Cref{prop:profunctor-comp} is functorial.
\end{lemma}
\begin{proof}
  It is easy to see that $\epcomp$ is well-defined on objects.
  Bifunctorality for morphisms then follows from bifunctorality of $\mone$ and functorality of coends.
\end{proof}

\begin{lemma}\label{lem:pcomp-ehom-unit}
  The functor $\ehom_\cat{M}$ of \Cref{prop:enriched-hom} is the left and right identity of the $\mone$-composition of \Cref{prop:profunctor-comp}.
\end{lemma} 
\begin{proof}
  The isomorphism on objects involves cocontinuity, the unitors of $\mone$, left Kan extending along the identity.
  Naturality is inherited from the naturality of the constructions involved.
  The unitors must also be monoidal natural transformations, which is true via cocontinuity and \Cref{lem:conc-str-mor-unr}.
\end{proof}

\begin{lemma}\label{lem:pcomp-assoc}
  The $\mone$-composition of \Cref{prop:profunctor-comp} is associative.
\end{lemma}
\begin{proof}
  The isomorphism uses cocontinuity and the associator of $\mone$.
  Naturality is inherited from the naturality of the constructions involved.
  The associator is a monoidal natural transformation by cocontinuity and \Cref{lem:conc-str-mor-assr}.
\end{proof}


\section{Proofs for change of enrichment}\label{sec:proofs:enrichment}

\begin{proof}[Proof of \Cref{thm:double-lax-change}]\label{prf:double-lax-change}
  \Cref{ax:mon.mor.id,ax:mon.mor.ass,ax:obj.lax.id,ax:obj.lax.ass} hold by the axioms for lax monoidal functors for the same reason lax monoidal functors preserve monoids.
  \Cref{ax:mon.mor.e.eta,ax:mon.mor.e.mu,ax:mon.mor.m.eta,ax:mon.mor.m.mu} each require the use of an axiom of double lax monoidal functors as shown below.
  \begin{align*}
    \mathsf{idt}_F . \epsilon &= F \mathsf{idt} . \eta_{\mone} . \epsilon \\
    &= F \mathsf{idt} . F \epsilon . \eta_{\mzero} \\
    &= F \mathsf{zero} . \eta_{\mzero} \\
    &= \mathsf{zero}_F
  \end{align*}
  \begin{align*}
    \mathsf{idt}_F . \nabla &= F \mathsf{idt} . \eta_{\mone} . \nabla \\
    &= F \mathsf{idt} . F \nabla . \mu_{\mzero} . (\eta_{\mone} \mzero \eta_{\mone}) \\
    &= F \mathsf{par} . F(\mathsf{idt} \mzero \mathsf{idt}) . \mu_{\mzero} . (\eta_{\mone} \mzero \eta_{\mone}) \\
    &= F \mathsf{par} . \mu_{\mzero} . (F\mathsf{idt} \mzero F\mathsf{idt}) . (\eta_{\mone} \mzero \eta_{\mone}) \\
    &= \mathsf{par}_F . ( \mathsf{idt}_F \mzero \mathsf{idt}_F)
  \end{align*}
  \begin{align*}
    \mathsf{seq}_F . ( \mathsf{zero}_F \mone \mathsf{zero}_F ). \Delta &= F \mathsf{seq} . \mu_{\mone} . ( F \mathsf{zero} \mone F \mathsf{zero} ) . (\eta_{\mzero} \mone \eta_{\mzero}) . \Delta \\
    &= F \mathsf{seq} . F( \mathsf{zero} \mone \mathsf{zero} ) . \mu_{\mone} . (\eta_{\mzero} \mone \eta_{\mzero}) . \Delta \\
    &= F \mathsf{seq} . F( \mathsf{zero} \mone \mathsf{zero} ) . F \Delta . \eta_{\mzero} \\
    &= F \mathsf{zero} . \eta_{\mzero} \\
    &= \mathsf{zero}_F
  \end{align*}
  \begin{align*}
    \mathsf{seq}_F . ( \mathsf{par}_F \mone \mathsf{par}_F ) . \zeta &= F \mathsf{seq} . \mu_{\mone} . ( F \mathsf{par} \mone F \mathsf{par} ) . (\mu_{\mzero} \mone \mu_{\mzero}) . \zeta \\
    &= F \mathsf{seq} . F( \mathsf{par} \mone \mathsf{par} ) . \mu_{\mone} . (\mu_{\mzero} \mone \mu_{\mzero}) . \zeta \\
    &= F \mathsf{seq} . F( \mathsf{par} \mone \mathsf{par} ) . F \zeta . \mu_{\mone} . (\mu_{\mzero} \mone \mu_{\mzero}) \\
    &= F \mathsf{par} . F( \mathsf{seq} \mzero \mathsf{seq} ) . \mu_{\mone} . (\mu_{\mzero} \mone \mu_{\mzero}) \\
    &= F \mathsf{par} . \mu_{\mone} . ( F \mathsf{seq} \mzero F \mathsf{seq} ) . (\mu_{\mzero} \mone \mu_{\mzero}) \\
    &= \mathsf{par}_F . ( \mathsf{seq}_F \mzero \mathsf{seq}_F )
  \end{align*}
  Similar checks show that $\overline{F}(G)$ is a $\cat{W}$-Freyd map.
  $\overline{F}$ is functorial by functorality of $F$.
\end{proof}

\end{document}